\documentclass[journal]{IEEEtran}
\usepackage{amsmath,amsfonts,amssymb,amsthm}
\usepackage{array}
\usepackage{textcomp}
\usepackage{stfloats}
\usepackage{url}
\usepackage{verbatim}
\usepackage{graphicx}
\usepackage{cite}

\usepackage{graphicx}
\usepackage{textcomp}
\usepackage{subfigure}

\usepackage{booktabs}
\usepackage[dvipsnames]{xcolor}
\usepackage{multirow}

\usepackage{threeparttable}
\usepackage[ruled,vlined]{algorithm2e}
\usepackage{balance}
\usepackage{lipsum}

\usepackage{csquotes}
\usepackage{makecell}
\usepackage[misc]{ifsym} 
\usepackage{url} 

\newtheorem{definition}{Definition}

\newtheorem{lemma}{Lemma}

\newtheorem{theorem}{Theorem}

\begin{document}

\title{PDET-LSH: Scalable In-Memory Indexing for High-Dimensional Approximate Nearest Neighbor Search with Quality Guarantees}

\author{Jiuqi Wei, Xiaodong Lee, Botao Peng, Quanqing Xu, Chuanhui Yang, Themis Palpanas
\thanks{This work is supported by EU Horizon projects TwinODIS ($101160009$), ARMADA
($101168951$), DataGEMS ($101188416$) and RECITALS ($101168490$), and by
$Y \Pi AI \Theta A$ \& NextGenerationEU project HARSH ($Y\Pi
3TA-0560901$) that is carried out within the framework of the National
Recovery and Resilience Plan “Greece 2.0” with funding from the European
Union – NextGenerationEU. 
\emph{(Corresponding author: Chuanhui Yang.)}}
\thanks{Jiuqi Wei, Quanqing Xu, and Chuanhui Yang are with Oceanbase, Ant Group, Beijing, China (e-mail: weijiuqi.wjq@antgroup.com; xuquanqing.xqq@oceanbase.com; rizhao.ych@oceanbase.com).}
\thanks{Xiaodong Lee and Botao Peng are with Institute of Computing Technology, Chinese Academy of Sciences, Beijing, China (e-mail: xl@ict.ac.cn; pengbotao@ict.ac.cn).}
\thanks{Themis Palpanas is with LIPADE, Universit{\'e} Paris Cit{\'e}, Paris, France (e-mail: themis@mi.parisdescartes.fr).}}


\markboth{Journal of \LaTeX\ Class Files,~Vol.~14, No.~8, August~2021}%
{Shell \MakeLowercase{\textit{et al.}}: A Sample Article Using IEEEtran.cls for IEEE Journals}


\maketitle

\begin{abstract}
Locality-sensitive hashing (LSH) is a well-known solution for approximate nearest neighbor (ANN) search with theoretical guarantees. 
Traditional LSH-based methods mainly focus on improving the efficiency and accuracy of query phase by designing different query strategies, but pay little attention to improving the efficiency of the indexing phase. 
They typically fine-tune existing data-oriented partitioning trees to index data points and support their query strategies.
However, their strategy to directly partition the multidimensional space is time-consuming, and performance degrades as the space dimensionality increases.
In this paper, we design an encoding-based tree called Dynamic Encoding Tree (DE-Tree) to improve the indexing efficiency and support efficient range queries. 
Based on DE-Tree, we propose a novel LSH scheme called DET-LSH. DET-LSH adopts a novel query strategy, 
which performs range queries in multiple independent index DE-Trees 
to reduce the probability of missing exact NN points. 
Extensive experiments demonstrate that while achieving best query accuracy, 
DET-LSH achieves up to 6x speedup in indexing time and 2x speedup in query time over the state-of-the-art LSH-based methods.
In addition, to further improve the performance of DET-LSH, we propose PDET-LSH, an in-memory method adopting the parallelization opportunities provided by multicore CPUs. 
PDET-LSH exhibits considerable advantages in indexing and query efficiency, especially on large-scale datasets.
Extensive experiments show that, while achieving the same query accuracy as DET-LSH, PDET-LSH offers up to 40x speedup in indexing time and 62x speedup in query answering time over the state-of-the-art LSH-based methods.
Our theoretical analysis demonstrates that DET-LSH and PDET-LSH offer probabilistic guarantees on query answering accuracy. 
\end{abstract}


\section{Introduction}
\noindent \textbf{Background and Problem.} Nearest neighbor (NN) search in high-dimensional Euclidean spaces is a fundamental problem in various fields, such as database, information retrieval, data mining, and machine learning. 
Given a dataset $\mathcal D$ of $n$ data points in $d$-dimensional space $\mathbb{R}^d$ and a query $q$, an NN query returns a point $o^* \in \mathcal D$ which has the minimum Euclidean distance to $q$ among all points in $\mathcal D$. 
However, NN search in high-dimensional datasets is challenging due to the \enquote{curse of dimensionality} phenomenon~\cite{weber1998quantitative}. 
In practice, Approximate Nearest Neighbor (ANN) search is often used as an alternative~\cite{hydra2}, 
sacrificing some query accuracy to achieve a huge improvement in efficiency~\cite{DBLP:journals/debu/00070P023,suco,iliassigmod25,leafi}. 
Given an approximation ratio $c$ and a query $q \in \mathbb{R}^d$, a $c$-ANN query returns a point $o$ whose distance to $q$ is at most $c$ times the distance between $q$ and its exact NN $o^*$, i.e., $\left\|q,o\right\| \leq c \cdot \left\|q,o^*\right\|$.

\textbf{Prior Work.} Locality-sensitive hashing (LSH)-based methods are known for their robust theoretical guarantees on the accuracy of query results~\cite{dblsh,lccslsh,pmlsh,detlsh}.
LSH is particularly suitable for scenarios where worst-case guarantees dominate average-case performance, as it provides distribution-independent sublinear query guarantees with provable correctness bounds~\cite{indyk1998approximate,andoni2008near}.
Moreover, unlike other ANN methods that rely on stable data distributions or trained routing structures~\cite{dong2011efficient,kraska2018case}, LSH remains applicable under severe distribution shift, since its guarantees depend solely on the underlying distance metric~\cite{andoni2015practical}, and its efficient random projections and lightweight index construction further enable rapid deployment and immediate query response in latency-sensitive applications such as LLM inference acceleration~\cite{chenmagicpig}.
LSH-based methods employ LSH functions to map high-dimensional points into lower-dimensional spaces for efficient indexing and querying. 
Due to LSH properties, nearby points in the original space have a higher probability of remaining close in the projected space~\cite{gionis1999similarity}. 
This allows for high-quality results by examining only the neighbors of the query in the projected space~\cite{datar2004locality}.
Following prior works, we could group LSH-based methods into three categories: 
1) boundary constraint (BC) based methods~\cite{e2lsh,lsbforest,sklsh,dblsh}; 
2) collision counting (C2) based methods ~\cite{c2lsh,qalsh,r2lsh,vhp,lccslsh}; 
and 3) distance metric (DM) based methods~\cite{srs,pmlsh}. 
BC methods map points into $L$ independent $K$-dimensional projected spaces, assigning each to a hash bucket bounded by a $K$-dimensional hypercube. 
Two points \emph{collide} if they land in the same bucket in any hash table.
Unlike BC, which requires collisions in all $K$ dimensions, C2 methods relax this condition, selecting points whose collision count with the query exceeds a threshold.
DM methods use the projected space distance to estimate the original distance with theoretical guarantees, selecting candidates via Euclidean range queries in that space.

\textbf{Limitations and Motivation.} 
Nowadays, new data is produced at an ever-increasing rate, and the size of datasets is continuously growing~\cite{wei2023data,wei2025dominate}. We need to manage large-scale data more efficiently to support further data analysis~\cite{hydra2}.
However, existing LSH-based methods mainly focus on reducing query time and improving query accuracy by designing different query strategies, but pay little attention to reducing indexing time. 
They typically fine-tune existing data-oriented partitioning trees to index points and support their query strategies, such as R*-Tree~\cite{rstartree} for DB-LSH~\cite{dblsh}, PM-Tree~\cite{pmtree} for PM-LSH~\cite{pmlsh}, and R-Tree~\cite{rtree} for SRS~\cite{srs}.
Data-oriented partitioning trees organize points hierarchically into bounding shapes (e.g., hyperrectangles)~\cite{rtree,mtree,rstartree,pmtree}, but partitioning in high-dimensional space is computationally expensive and becomes less effective as dimensionality increases~\cite{weber1998quantitative}. 
These limitations constrain the efficiency and dimensionality of the projected space, highlighting the need for a more efficient tree structure.
From another perspective, building a more efficient tree structure boosts query accuracy by enabling more trees to be constructed within the same indexing time. 
For instance, DB-LSH~\cite{dblsh} improves result accuracy by using multiple R*-Trees to minimize missing true nearest neighbors. In addition, efficient index construction and query processing are essential for large-scale ANN search.
Many approaches~\cite{messi,paris+,sing,echihabi2022hercules,dumpyos,azizi2023elpis,chatzakis2023odyssey,fatourou2023fresh} take advantage of the parallelization opportunities (i.e., SIMD instructions, multi-socket, and multi-core architectures) provided by modern hardware to accelerate indexing and querying. 
However, none of the state-of-the-art LSH-based methods~\cite{pmlsh,dblsh,lccslsh} offer parallel algorithms to improve the efficiency of indexing and querying, which limits their competitiveness and suitability for demanding modern applications.

\textbf{Our Method.} 
In this paper, we propose a novel tree structure called Dynamic Encoding Tree (DE-Tree), a novel LSH scheme called DET-LSH, and a parallel version of DET-LSH called PDET-LSH to address the high-dimensional $c$-ANN search problem more efficiently and accurately than current methods.
First, we present DE-Tree, an encoding-based tree that independently divides and encodes each dimension of the projected space according to the data distribution.
Unlike data-oriented trees, it avoids direct multi-dimensional partitioning, improving indexing efficiency. 
Its dynamic encoding ensures nearby points have similar encoding representations, enhancing query accuracy. 
DE-Tree also supports efficient range queries because the upper and lower bound distances between a query point and any DE-Tree node can be easily calculated.
Second, we propose DET-LSH, a novel LSH scheme that dynamically encodes $K$-dimensional projected points and builds $L$ DE-Trees from the encoded data. 
DET-LSH employs a two-step query strategy combining BC and DM principles: first, coarse-grained range queries in DE-Trees retrieve a candidate set; second, fine-grained exact distance calculations are performed to sort and return the final results. 
This approach uses coarse filtering to boost efficiency and fine-grained verification to ensure accuracy.
Third, we conduct a rigorous theoretical analysis showing that DET-LSH can correctly answer a $c^2$-$k$-ANN query with a constant probability. 
Fourth, we redesign the encoding, indexing, and query algorithms of DET-LSH, and propose a new method called PDET-LSH, which takes advantage of the parallelization opportunities provided by multicore CPUs to accelerate indexing and querying.
Furthermore, extensive experiments on various real-world datasets show that DET-LSH and PDET-LSH outperform existing LSH-based methods in both efficiency and accuracy.

Our main contributions are summarized as follows\footnote{A preliminary version of this paper has appeared elsewhere~\cite{detlsh}.}.

\begin{itemize}
	\item We present a novel encoding-based tree structure called DE-Tree. Compared with data-oriented partitioning trees used in existing LSH-based methods, DE-Tree has better indexing efficiency and can support more efficient range queries based on the Euclidean distance metric.
	\item We propose DET-LSH, a novel LSH scheme based on DE-Tree, and a novel query strategy that takes into account both efficiency and accuracy. We provide a theoretical analysis showing that DET-LSH answers a $c^2$-$k$-ANN query with a constant success probability.
    \item We propose PDET-LSH, which takes advantage of the parallelization opportunities provided by multicore CPUs to accelerate indexing and querying. To the best of our knowledge, PDET-LSH is the first parallel solution among the state-of-the-art LSH-based solutions. Our design of PDET-LSH could instigate similar work on other LSH-based solutions. 
	\item We conduct extensive experiments, demonstrating that DET-LSH and PDET-LSH can achieve better efficiency and accuracy than existing LSH-based methods. While achieving better query accuracy than competitors, DET-LSH achieves up to 6x speedup in indexing time and 2x speedup in query time over the state-of-the-art LSH-based methods. PDET-LSH achieves the same query accuracy as DET-LSH, and up to 40x speedup in indexing time and 62x speedup in query time over the state-of-the-art LSH-based methods.
\end{itemize}

\section{Preliminaries} \label{chapter3}

\subsection{Problem Definition}  \label{chapter3.1}

Let $\mathcal D$ be a dataset of points in $d$-dimensional space $\mathbb{R}^d$. 
The dataset cardinality is denoted as $\lvert \mathcal D \rvert=n$, 
and let $\left\|o_1,o_2\right\|$ denote the distance between points $o_1,o_2\in \mathcal D$. 
The query point $q \in \mathbb{R}^d$.

\begin{definition}[$c$-ANN]\label{def1}
	Given a query point $q$ and an approximation ratio $c > 1$, let $o^*$ be the exact nearest neighbor of $q$ in $\mathcal D$. A $c$-ANN query returns a point $o \in \mathcal D$ satisfying $\left\|q,o\right\| \leq c \cdot \left\|q,o^*\right\|$.
\end{definition} 

The $c$-ANN query can be generalized to $c$-$k$-ANN query that returns $k$ approximate nearest points, where $k$ is a positive integer.

\begin{definition}[$c$-$k$-ANN]\label{def2}
	Given a query point $q$, an approximation ratio $c > 1$, and an integer $k$. Let $o^*_i$ be the $i$-th exact nearest neighbor of $q$ in $\mathcal D$. A $c$-$k$-ANN query returns $k$ points $o_1,o_2,...,o_k$. For each $o_i \in D$ satisfying $\left\|q,o_i\right\| \leq c \cdot \left\|q,o^*_i\right\|$, where $i \in [1,k]$.
\end{definition}  

In fact, LSH-based methods do not solve $c$-ANN queries directly because $o^*$ and $\left\|q,o^*\right\|$ is not known in advance~\cite{lccslsh,pmlsh,dblsh}. 
Instead, they solve the problem of ($r$,$c$)-ANN proposed in~\cite{indyk1998approximate}.

\begin{definition}[($r$,$c$)-ANN]\label{def3}
	Given a query point $q$, an approximation ratio $c > 1$, and a search radius $r$. An ($r$,$c$)-ANN query returns the following result:
	\begin{enumerate}
		\item If there exists a point $o \in \mathcal D$ such that $\left\|q,o\right\| \leq r$, then return a point $o^{\prime} \in \mathcal D$ such that $\left\|q,o^{\prime}\right\| \leq c \cdot r$;
		\item If for all $o \in \mathcal D$ we have $\left\|q,o\right\| > c \cdot r$, then return nothing;
		\item If for the point $o$ closest to $q$ we have $r < \left\|q,o\right\| \leq c \cdot r$, then return $o$ or nothing. 
	\end{enumerate}
\end{definition}

\begin{table}
	\centering
	\caption{Notations}
	\label{table1}
 {
	\begin{tabular}{cc}
		\toprule
		\textbf{Notation} & \textbf{Description} \\
		\midrule
		$\mathbb{R}^d$ & $d$-dimensional Euclidean space \\
		$\mathcal D$ & Dataset of points in $\mathbb{R}^d$ \\
		$n$ & Dataset cardinality  $\lvert \mathcal D \rvert$ \\
		$o$ & A data point in $\mathcal D$ \\
		$q$ & A query point in $\mathbb{R}^d$ \\
		$o^{\prime},q^{\prime}$ & $o$ and $q$ in the projected space \\
		$o^*,o^*_i$ & The first and $i$-th nearest point in $\mathcal D$ to $q$ \\
		$\left\|o_1,o_2\right\|$ & The Euclidean distance between $o_1$ and $o_2$ \\
		$s,s^\prime$ & Abbreviation for $\left\|o_1,o_2\right\|$ and  $\left\|o_1^\prime,o_2^\prime \right\|$ \\
		$h(o)$ & Hash function \\
		$\mathcal H (o)$ & $[h_1(o),...,h_K(o)]$, the coordinates of $o^{\prime}$ \\
		$\mathcal H_i(o)$ & Coordinates of $o^{\prime}$ in the $i$-th project space\\
		$c$ & Approximation ratio \\
		$r$ & Search radius in the original space \\
		$r_{min}$ & The initial search radius \\
		$d$ & Dimension of the original space \\
		$K$ & Dimension of the projected space \\
		$L$ & Number of independent projected spaces \\
		$\beta$ & Maximum false positive percentage \\
        $N_r$ & Number of regions in each projected space\\
        $N_w$ & Number of parallel workers\\
        $N_q$ & Number of queues during query answering\\
        $S_p$ & Speedup ratio\\
		\bottomrule
	\end{tabular}
 } 
\end{table}

The $c$-ANN query can be transformed into a series of ($r$,$c$)-ANN queries with increasing radii until a point is returned. 
The search radius $r$ is continuously enlarged by multiplying $c$, i.e., $r=r_{min}, r_{min} \cdot c, r_{min} \cdot c^2,...$, where $r_{min}$ is the initial search radius. In this way, as proven by~\cite{indyk1998approximate}, the ANN query can be answered with an approximation ratio $c^2$, i.e., $c^2$-ANN.

\subsection{Locality-Sensitive Hashing}  \label{chapter3.2}

The capability of an LSH function $h$ is to project closer data points into the same hash bucket with a higher probability, i.e., $h(o_1)=h(o_2)$. Formally, the definition of LSH used in Euclidean space is given below~\cite{pmlsh,dblsh}: 

\begin{definition}[LSH]\label{def4}
	Given a distance $r$, an approximation ratio $c>1$, a family of hash functions $\mathcal H = \{h:\mathbb{R}^d \rightarrow \mathbb{R}\}$ is called ($r$,$cr$,$p_1$,$p_2$)-locality-sensitive, if for $\forall o_1,o_2 \in \mathbb{R}^d$, it satisfies both of the following conditions: 
	\begin{enumerate}
	   	\item If $\left\|o_1,o_2\right\| \leq r$, $\Pr{[h(o_1)=h(o_2)]} \geq p_1$;
	   	\item If $\left\|o_1,o_2\right\| > cr$, $\Pr{[h(o_1)=h(o_2)]} \leq p_2$,
	\end{enumerate}
	where $h \in \mathcal H$ is randomly chosen, and the probability values $p_1$ and $p_2$ satisfy $p_1>p_2$.
\end{definition} 

A widely adopted LSH family for the Euclidean space is defined as follows~\cite{qalsh}:

\begin{equation}  \label{eq1}
	h(o)=\vec{a} \cdot \vec{o},
\end{equation}

\noindent where $\vec{o}$ is the vector representation of a point $o \in \mathbb{R}^d$ and $\vec{a}$ is a $d$-dimensional vector where each entry is independently chosen from the standard normal distribution $\mathcal N(0,1)$.

\subsection{$p$-Stable Distribution and $\chi^2$ Distribution}  \label{chapter3.3}

A distribution $\mathcal T$ is called $p$-stable, if for any $u$ real numbers $v_1,...,v_u$ and identically distributed (i.i.d.) variables $X_1,...,X_u$ following $\mathcal T$ distribution, $\sum_{i=1}^{u}v_iX_i$ has the same distribution as $(\sum_{i=1}^{u}\lvert v_i \rvert ^p)^{1/p} \cdot X$, where $X$ is a random variable with distribution $\mathcal T$~\cite{datar2004locality}. $p$-stable distribution exists for any $p \in (0,2]$~\cite{zolotarev1986one}, and $\mathcal T$ is the normal distribution when $p=2$. 

Let $o^{\prime}= \mathcal H (o)=  [h_1(o),...,h_K(o)]$ denote the point $o$ in the $K$-dimensional projected space. For any two points $o_1,o_2 \in \mathcal D$, let $s=\left\|o_1,o_2\right\|$ and $s^\prime=\left\|o_1^\prime,o_2^\prime\right\|$ denote the Euclidean distances between $o_1$ and $o_2$ in the original space and in the projected space.

\begin{lemma}\label{lemma1}
	$\frac{s^{\prime2}}{s^2}$ follows the $\chi^2(K)$ distribution.
\end{lemma} 

\begin{proof}
	Let $h^\prime=h(o_1)-h(o_2)=\vec{a} \cdot (\vec{o_1}-\vec{o_2})=\sum_{i=1}^{d}(o_1[i]-o_2[i]) \cdot a[i]$, where $a[i]$ follows the $\mathcal N(0,1)$ distribution. Since 2-stable distribution is the normal distribution, $h^\prime$ has the same distribution as $(\sum_{i=1}^{d} (o_1[i]-o_2[i])^2)^{1/2} \cdot X=s \cdot X$, where $X$ is a random variable with distribution $\mathcal N(0,1)$. Therefore $\frac{h^\prime}{s}$ follows the $\mathcal N(0,1)$ distribution. Given $K$ hash functions $h_1(\cdot),...,h_K(\cdot)$, we have $\frac{h_1^{\prime2}+...+h_K^{\prime2}}{s^2}=\frac{s^{\prime2}}{s^2}$, which has the same distribution as $\sum_{i=1}^{K}X_i^2$. Thus, $\frac{s^{\prime2}}{s^2}$ follows the $\chi^2(K)$ distribution.
\end{proof}

\begin{lemma}\label{lemma2}
	Given $s$ and $s^\prime$ we have:
	\begin{equation}  \label{eqkafang}
		\Pr{[s^\prime>s\sqrt{\chi^2_{\alpha}(K)}]} = \alpha,
	\end{equation}
	\noindent where $\chi^2_\alpha(K)$ is the upper quantile of a distribution $Y \sim \chi^2(K)$, i.e., $\Pr{[Y > \chi^2_\alpha(K)]} = \alpha$.
\end{lemma} 

\begin{proof}
	From Lemma~\ref{lemma1}, we have 	$\frac{s^{\prime2}}{s^2} \sim \chi^2(K)$. Since $\chi^2_{\alpha}(K)$ is the $\alpha$ upper quantiles of $\chi^2(K)$ distribution, we have $\Pr{[ \frac{s^{\prime2}}{s^2} > \chi^2_{\alpha}(K)]} = \alpha$. Transform the formulas, we have $\Pr{[s^\prime>s\sqrt{\chi^2_{\alpha}(K)}]} = \alpha$.
\end{proof}

\begin{figure*} [tb]
\vspace*{-0.4cm}
	\flushleft 
	\subfigcapskip=5pt
	\subfigure[Encoding phase and indexing phase.]{
		\includegraphics[width=0.47\linewidth]{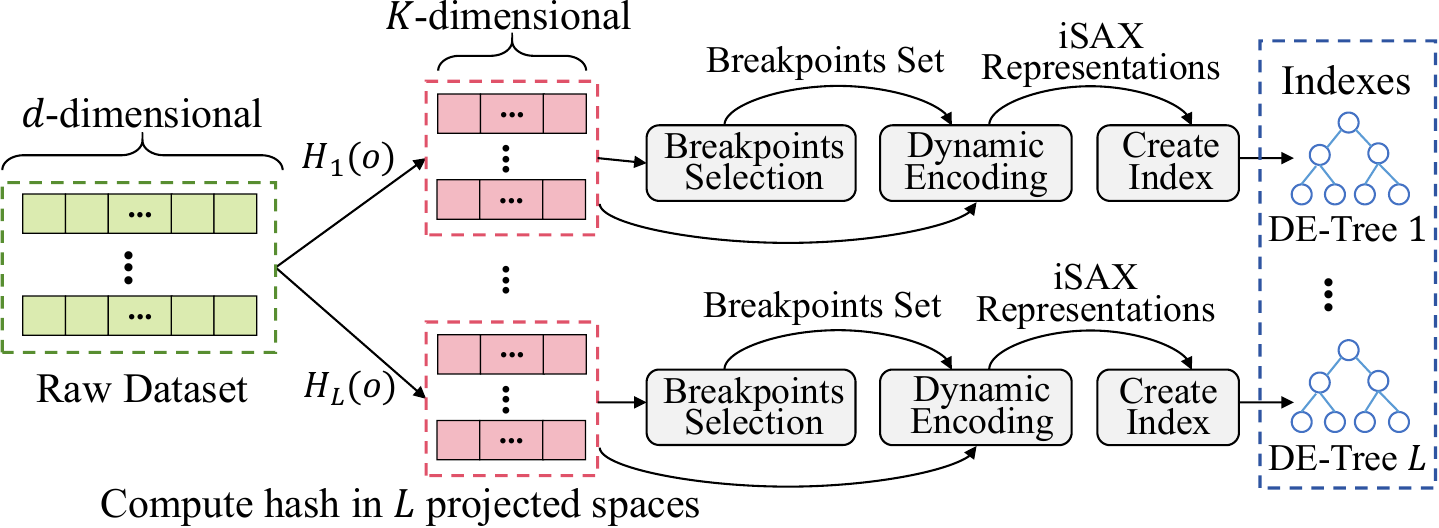}
		\label{overview1}}\hspace{3mm}
	\subfigure[Query phase.]{
		\includegraphics[width=0.47\linewidth]{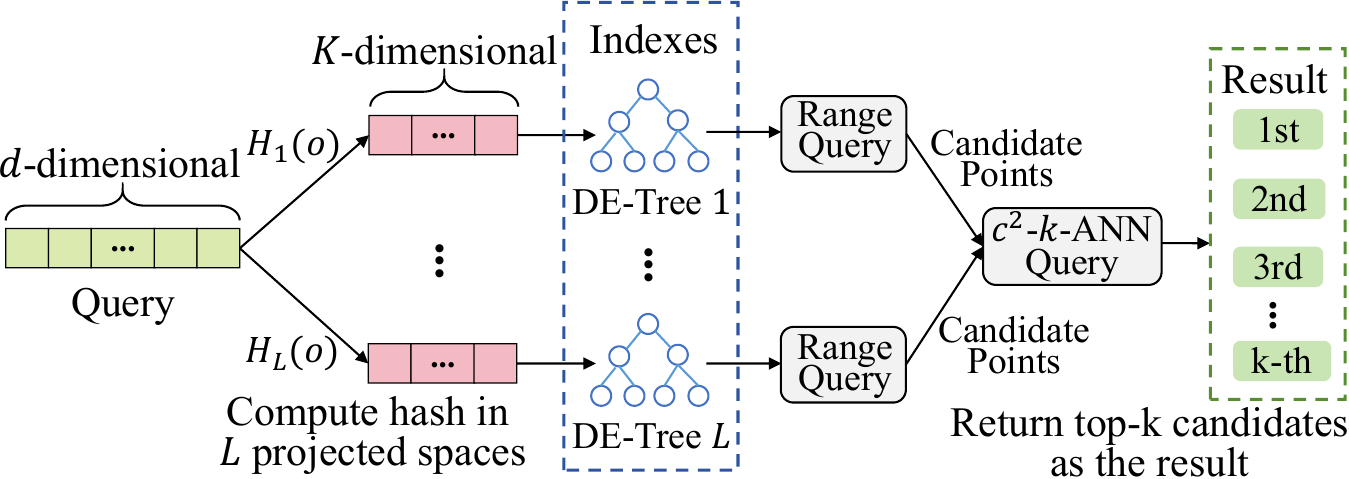}
		\label{overview2}}
	\caption{Overview of the DET-LSH workflow.}
	\label{overview}
\end{figure*}

\section{The DET-LSH Method} \label{chapter4}

In this section, we present the details of DET-LSH and the design of Dynamic Encoding Tree (DE-Tree). 
DET-LSH consists of three phases: an encoding phase to encode the LSH-based projected points into iSAX representations; 
an indexing phase to construct DE-Trees based on the iSAX representations; 
a query phase to perform range queries in DE-Trees for ANN search. 

Figure~\ref{overview} provides a high-level overview of the workflow for DET-LSH.
In the encoding phase, $K \cdot L$ hash functions project the raw data into $L$ independent $K$-dimensional spaces, where data-driven breakpoints are dynamically selected to partition each dimension and encode projected points into iSAX representations.
During indexing, a DE-Tree is constructed in each projected space via top-down node splitting based on the iSAX codes.
At query time, the query is mapped to the same $L$ projected spaces, range queries are performed over multiple DE-Trees using their upper and lower bound distance properties to generate a candidate set, and the final $c^2$-$k$-ANN results are returned according to the proposed query strategy.

\begin{algorithm}[tb]
\small

	\caption{Dynamic Encoding}                                                                     
	\label{dynamic_encoding}
	\LinesNumbered
	\KwIn{Parameters $K$, $L$, $n$, all points in projected spaces $P$, sample size $n_s$, number of regions in each projected space $N_r$}
	\KwOut{A set of encoded points $EP$}
	Initialize $EP$ with size $n \cdot L \cdot K$; \\
    $B \leftarrow$ Select breakpoints by running multiple rounds of the \textit{QuickSelect} algorithm combined with the \textit{divide-and-conquer} strategy; \\
	\For{$i=1$ to $L$}{
		\For{$j=1$ to $K$}{
			\For{$z=1$ to $n$}{
                Obtain $o_z$ from $P$; \\
				Use \textit{BinarySearch} to find integer $b \in [1,N_r]$ such that $B_{ij}(b) \leq h_{ij}(o_z) \leq B_{ij}(b+1)$; \\
				$EP_{ij}(o_z) \leftarrow b$-th symbol in the 8-bit alphabet; \\			
			}
		}
	}
	\Return $EP$; \\
\end{algorithm}

\subsection{Encoding Phase} \label{Encoding Phase}

Before indexing projected points with DE-Trees, DET-LSH encodes them into indexable Symbolic Aggregate approXimation (iSAX) representations~\cite{isax}. 
iSAX partitions each dimension using a set of non-uniform breakpoints and assigns a bit-wise symbol to each resulting region, such that points falling in the same region share identical iSAX codes. 
As illustrated in Figure~\ref{isaxencoding}, using three breakpoints per dimension in a two-dimensional space yields four regions per dimension, represented by 2-bit symbols, and thus 16 regions in total. 
An index constructed over these representations is shown in Figure~\ref{isaxindex}. 
In practice, high-quality approximations can be achieved with at most 256 symbols per dimension, allowing each dimension to be encoded using an 8-bit alphabet~\cite{iSAX2}.

This design benefits LSH-based ANN search by independently encoding each dimension, avoiding expensive multi-dimensional partitioning and thus improving indexing efficiency. Moreover, the symbolic regions enable efficient upper and lower bound distance computation for effective pruning during range queries, enhancing scalability and query performance in high-dimensional spaces.

\noindent \textbf{Static encoding scheme.} 
In data series similarity search, traditional iSAX-based methods adopt a static encoding scheme~\cite{messi,paris+,sing}. 
Exploiting the approximately Gaussian distribution of normalized series~\cite{isax}, breakpoints $b_1,\ldots,b_{a-1}$ are chosen such that each interval under the $\mathcal{N}(0,1)$ curve has equal probability mass $\frac{1}{a}$, with $b_0=-\infty$ and $b_a=+\infty$, yielding dataset-independent encodings. 
Existing methods follow this static scheme by mapping each coordinate to its enclosing breakpoint interval via a lookup table. 
However, since ANN datasets often exhibit arbitrary and non-Gaussian distributions, such static encoding becomes unsuitable.

\noindent \textbf{Dynamic encoding scheme.}
DET-LSH adopts a dynamic encoding scheme that selects breakpoints according to the data distribution, aiming to evenly partition points such that each region contains approximately the same number of points. 
Given a dataset of cardinality $n$, we first apply $K\!\cdot\!L$ hash functions to obtain $K$-dimensional representations in $L$ projected spaces, where $\mathcal H_i(o)=[h_{i1}(o),\ldots,h_{iK}(o)]$. 
For each projected space $i$ and dimension $j$, let $C_{ij}=[h_{ij}(o_1),\ldots,h_{ij}(o_n)]$ denote the set of coordinates, and let $C^{\uparrow}_{ij}$ be its sorted version in ascending order. 
To divide each dimension into $N_r=256$ regions, we select breakpoints $B_{ij}$ such that 
$B_{ij}(z)=C^{\uparrow}_{ij}(\lfloor n/N_r \rfloor \cdot (z-1))$ for $z=2,\ldots,N_r$, with boundary conditions $B_{ij}(1)=C^{\uparrow}_{ij}(1)$ and $B_{ij}(N_r+1)=C^{\uparrow}_{ij}(n)$. 
Thus, for each dimension in each projected space, $N_r\!+\!1$ breakpoints are determined dynamically, and each coordinate $h_{ij}(o)$ is independently encoded according to its corresponding breakpoints $B_{ij}$.

\begin{algorithm}[tb]
\small

	\caption{Create Index}                                                                        
	\label{create_index}
	\LinesNumbered
	\KwIn{Parameters $K$, $L$, $n$, encoded points set $EP$, maximum size of a leaf node $max\_size$}
	\KwOut{A set of DE-Trees: $DETs=[T_1,...,T_L]$}
	\For{$i=1$ to $L$}{
		Initialize $T_i$ and generate $2^{K}$ first-layer nodes as the original leaf nodes; \\
		\For{$z=1$ to $n$}{
			$ep_i(o_z) \leftarrow (EP_{i1}(o_z),...,EP_{iK}(o_z))$; \\
			$pos_z \leftarrow$ the position of $o_z$ in the dataset; \\
			$target\_leaf \leftarrow$ leaf node of $T_i$ to insert $\langle ep_i(o_z),pos_z \rangle$; \\
			\While{$sizeof(target\_leaf) \geq max\_size$}{
				SplitNode($target\_leaf$); \\
				$target\_leaf \leftarrow$ the new leaf node to insert $\langle ep_i(o_z),pos_z \rangle$; \\
			}
			Insert $\langle ep_i(o_z),pos_z \rangle$ to $target\_leaf$; \\
		}
	}
	\Return $DETs$; \\
\end{algorithm}

A straightforward approach is to fully sort $C_{ij}$ to obtain $C^{\uparrow}_{ij}$ and then extract the desired breakpoints. 
However, since only $N_r\!+\!1$ order statistics are required, complete sorting is wasteful. 
Accordingly, we design a dynamic encoding algorithm (Algorithm~\ref{dynamic_encoding}) that operates directly on the unordered $C_{ij}$. 
As illustrated in Figure~\ref{breakpoints}, breakpoints are obtained through multiple rounds of \emph{QuickSelect} within the \emph{divide-and-conquer} strategy. 
To further improve efficiency, we randomly sample $n_s$ points from the dataset and compute breakpoints on the sampled set, with $n_s$ set to $0.1n$ in practice.
For an unordered set $C_{ij}$, \emph{QuickSelect}$(start,q,end)$ identifies the $q$-th smallest element within $[start,end]$ and places it at position $start+q$, such that elements to its left (right) are smaller (larger). 
Thus, a single breakpoint can be obtained with one \emph{QuickSelect} call. 
With $N_r=256$, we apply a \emph{divide-and-conquer} strategy: $\log_2 N_r$ rounds are performed, where the $z$-th round selects $2^{z-1}$ breakpoints by running \emph{QuickSelect} on the subregions generated in the previous round ($z=1,\ldots,\log_2 N_r$). 
For each $C_{ij}$, the minimum and maximum elements are taken as $B_{ij}(1)$ and $B_{ij}(N_r\!+\!1)$, respectively. 
This strategy yields a $3\times$ speedup over full sorting (Section~\ref{selfevaluation}). 
Given the breakpoint set $B$, Algorithm~\ref{dynamic_encoding} encodes all points into iSAX representations and outputs the encoded dataset (lines~3-8).

\begin{figure}[tb] 
	\centering
	\includegraphics[width=0.85\linewidth]{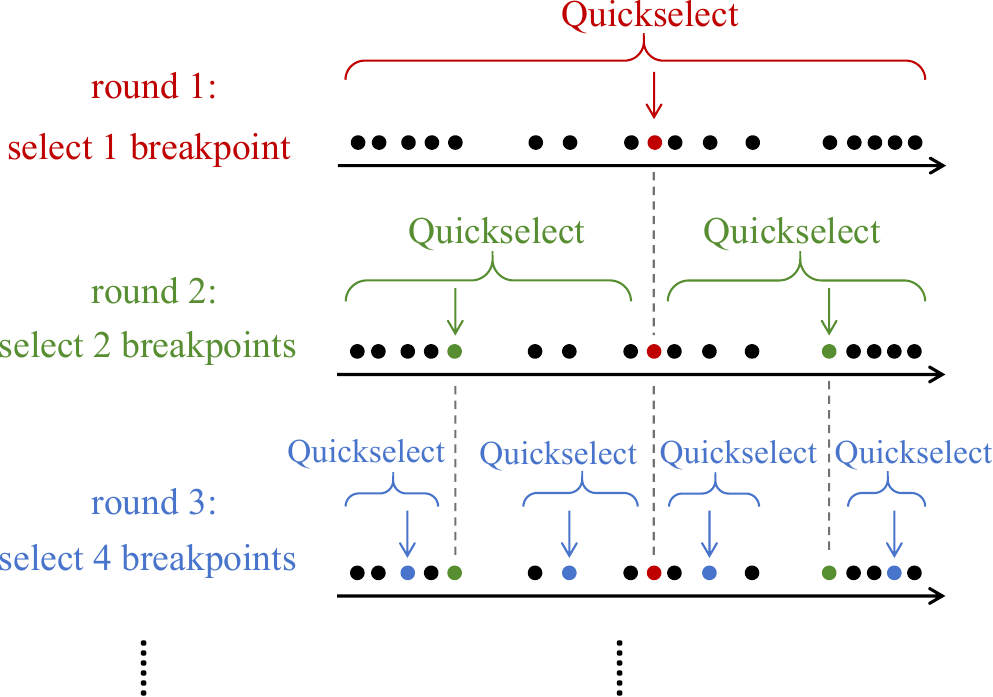}
	\caption{Illustration of the breakpoints selection process.}
	\label{breakpoints}
\end{figure}

\subsection{Indexing Phase} \label{Indexing Phase}

As discussed above, DET-LSH constructs $L$ DE-Trees to support query processing. 
Algorithm~\ref{create_index} outlines the construction of these trees from the encoded point set $EP$. 
For each DE-Tree, construction begins by initializing the first-layer nodes as the root’s children (line~2). 
According to the iSAX encoding rules, each dimension is initially split into two cases, $0^*$ and $1^*$, yielding $2^K$ first-layer nodes per tree (Figure~\ref{isaxindex}). 
For each data point $o_z$, its encoded representation $ep_i(o_z)$ in the $i$-th DE-Tree $T_i$ and its dataset position $pos_z$ are obtained (lines~3-5), and the pair $\langle ep_i(o_z),pos_z \rangle$ is inserted into the corresponding leaf node (line~6). 
If the target leaf node overflows, it is recursively split until insertion succeeds (lines~7-10). 
Only leaf nodes store point information (encoded representations and positions), whereas internal nodes maintain solely index metadata.

In a DE-Tree, except the root node that has $2^K$ children, other internal nodes have only two children.  
This is because node splitting performs a binary refinement on a single selected dimension among the $K$ dimensions. 
As illustrated in Figure~\ref{isaxindex}, splitting the node $[0^*,0^*]$ along the first dimension yields two children with representations $[00,0^*]$ and $[01,0^*]$. 
The choice of splitting dimension is critical: to promote balanced trees, we select the dimension that most evenly partitions the points between the two child nodes.

\begin{figure} [t!]
	\centering
	\subfigcapskip=5pt
	\subfigure[Encode data points into iSAX representations.]{
		\includegraphics[width=0.75\linewidth]{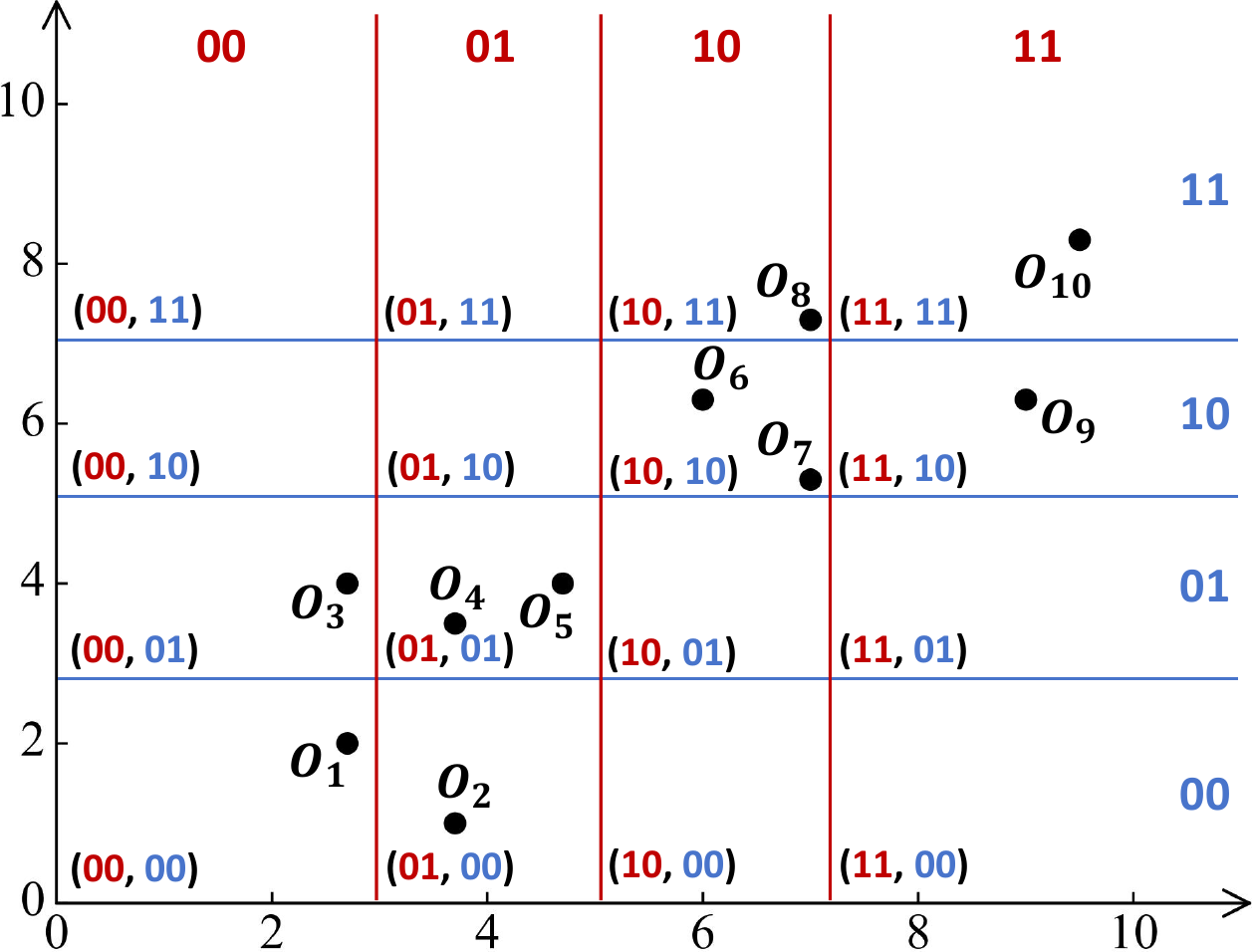}
		\label{isaxencoding}}
	\subfigure[An index based on the iSAX representations.]{
		\includegraphics[width=0.75\linewidth]{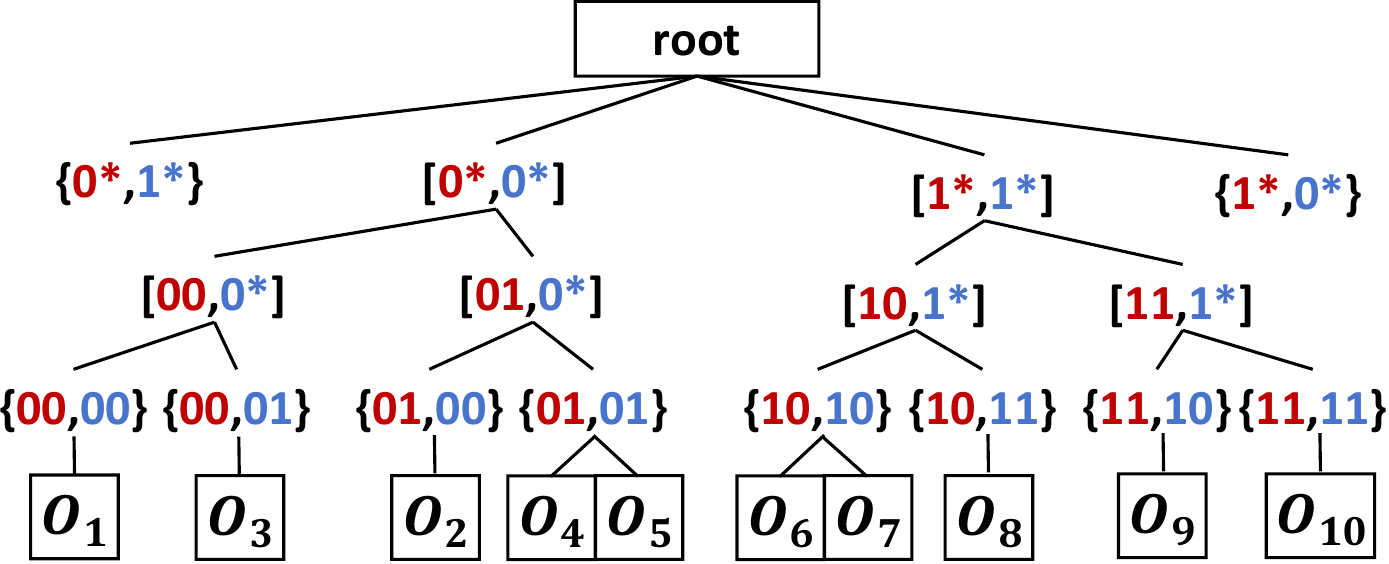}
		\label{isaxindex}}
	\caption{Illustration of the Dynamic Encoding Tree (DE-Tree).}
	\label{Encoding-based Trees}
\end{figure}

\begin{algorithm}[tb]
\small

	\caption{DET Range Query}                                                                   
	\label{range_query}
	\LinesNumbered
	\KwIn{A projected query point $q^\prime$, the search radius $r^\prime$, the index DE-Tree $T$, project dimension $K$}
	\KwOut{A set of points $S$}
	Initialize a points set $S \leftarrow \varnothing$; \\
	\For{$i=1$ to $2^K$}{
		$node \leftarrow$ the $i$-th child of root node in $T$; \\
        Traverse the subtree rooted at $node$ and add to $S$ the data points in its leaf nodes whose distance to $q^\prime$ is smaller than the search radius $r^\prime$; \\
	}
	\Return $S$;
\end{algorithm}

This design of DE-Tree motivates performing a range query before the kNN search. 
The range query efficiently prunes the search space using the computed lower and upper bound distances, thereby restricting the search to a much smaller candidate set. 
The subsequent kNN query is then executed only within this refined set, reducing unnecessary distance computations and improving overall query efficiency.

\begin{algorithm}[tb]
	\caption{($r$,$c$)-ANN Query}   
 \small
 
	\label{rcann}
	\LinesNumbered
	\KwIn{A query point $q$, parameters $K$, $L$, $n$, $c$, $r$, $\epsilon$, $\beta$, index DE-Trees $DETs=[T_1,...,T_L]$}
	\KwOut{A point $o$ or $\varnothing$}
	Initialize a candidate set $S \leftarrow \varnothing$; \\
	\For{$i=1$ to $L$}{
		Compute $q_i^\prime=H _i(q)=[h_{i1}(q),...,h_{iK}(q)]$; \\
		$S_i \leftarrow$ \textbf{call} DETRangeQuery($q_i^\prime,\epsilon \cdot r, T_i, K$); \\
		$S \leftarrow S \cup S_i$; \\
		\If{$\lvert S \rvert \geq \beta n+1$}{
			\Return the point $o$ closest to $q$ in $S$; \\
		}
	}
	\If{$\lvert \left\{ o \mid o \in S \land \left\|o,q\right\| \leq c \cdot r \right\} \rvert \geq 1$}{
		\Return the point $o$ closest to $q$ in $S$; \\
	}
	\Return $\varnothing$;
\end{algorithm}

\begin{algorithm}[ht]
	\caption{$c^2$-$k$-ANN Query}               
 \small

	\label{ckann}
	\LinesNumbered
	\KwIn{A query point $q$, parameters $K$, $L$, $n$, $c$, $r_{min}$, $\epsilon$, $\beta$, $k$, index DE-Trees $DETs=[T_1,...,T_L]$}
	\KwOut{$k$ nearest points to $q$ in $S$}
	Initialize a candidate set $S \leftarrow \varnothing$ and set $r \leftarrow r_{min}$; \\
	\While{\textit{TRUE}}{
		\For{$i=1$ to $L$}{
			Compute $q_i^\prime=H _i(q)=[h_{i1}(q),...,h_{iK}(q)]$; \\
			$S_i \leftarrow$ \textbf{call} DETRangeQuery($q_i^\prime,\epsilon \cdot r, T_i, K$); \\
			$S \leftarrow S \cup S_i$; \\
			\If{$\lvert S \rvert \geq \beta n+k$}{
				\Return the \textit{top}-$k$ points closest to $q$ in $S$; \\
			}
		}
		\If{$\lvert \left\{ o \mid o \in S \land \left\|o,q\right\| \leq c \cdot r \right\} \rvert \geq k$}{
			\Return the \textit{top}-$k$ points closest to $q$ in $S$; \\
		}
		$r \leftarrow c \cdot r$;
	}
\end{algorithm}

\noindent \textbf{DET Range Query.} 
Algorithm~\ref{range_query} performs range queries on a DE-Tree by traversing subtrees to identify leaf nodes containing points within the search radius $r'$. 
The traversal starts from all $2^K$ children of the root as entry points and proceeds recursively (lines~2-4). 
During traversal, nodes are pruned using lower and upper bound distances to the query point $q'$. 
If a node’s lower bound distance exceeds $r'$, all points in its subtree can be safely discarded. 
For a leaf node, if its upper bound distance is no greater than $r'$, all contained points are directly added to the candidate set $S$; if $r'$ lies between the lower and upper bounds, only points within distance $r'$ are examined and added. 
Non-leaf nodes that cannot be pruned are further expanded by traversing their subtrees.

\subsection{Query Phase} \label{Query Phase}

\noindent \textbf{($r$,$c$)-ANN Query.} 
Algorithm~\ref{rcann} describes how DET-LSH answers an $(r,c)$-ANN query for an arbitrary search radius $r$. 
After indexing, DET-LSH maintains $L$ DE-Trees $T_1,\ldots,T_L$. 
Given a query $q$, we process the $L$ projected spaces sequentially: for the $i$-th space, the projected query $q_i'$ is first computed (line~3), and a range query is executed on $T_i$ using Algorithm~\ref{range_query} with search radius $\epsilon r$ (line~4). 
The parameter $\epsilon$ ensures that if $\|o-q\|\le r$, then $\|o'-q'\|\le \epsilon r$ holds with constant probability, as formally analyzed in Lemma~\ref{lemma3}. 
Candidates returned from each range query are accumulated in a set $S$ (line~5). 
Once $|S| \geq \beta n+1$, the point closest to $q$ is returned, where $\beta$ is the maximum false positive percentage (lines~6-7). 
After all $L$ trees are processed, if $|S|<\beta n+1$ and $S$ contains at least one point within distance $c \cdot r$ from $q$, the nearest such point is returned (lines~8-9); otherwise, the algorithm returns null (line~10). 
According to Theorem~\ref{theorem1}, DET-LSH can correctly answer an ($r$,$c$)-ANN query with a constant probability.

\noindent \textbf{$c^2$-$k$-ANN Query.}
Since the true nearest neighbor $o^*$ and its distance $\|q,o^*\|$ are unknown in advance, a $c$-$k$-ANN query cannot be executed with a fixed radius $r$ as in ($r$,$c$)-ANN. 
Instead, DET-LSH issues a sequence of ($r$,$c$)-ANN queries with progressively increasing radii until sufficient candidates are obtained. 
Algorithm~\ref{ckann} details this procedure: most steps (lines~3-10) mirror Algorithm~\ref{rcann} (lines~2-9), with the key difference that termination and return conditions explicitly account for $k$. 
When neither termination condition is met, the search radius is enlarged for the next iteration (line~11). 
According to Theorem~\ref{theorem2}, DET-LSH can correctly answer a $c^2$-$k$-ANN query with a constant probability.

\begin{figure}[tb] 
	\centering
	\includegraphics[width=\linewidth]{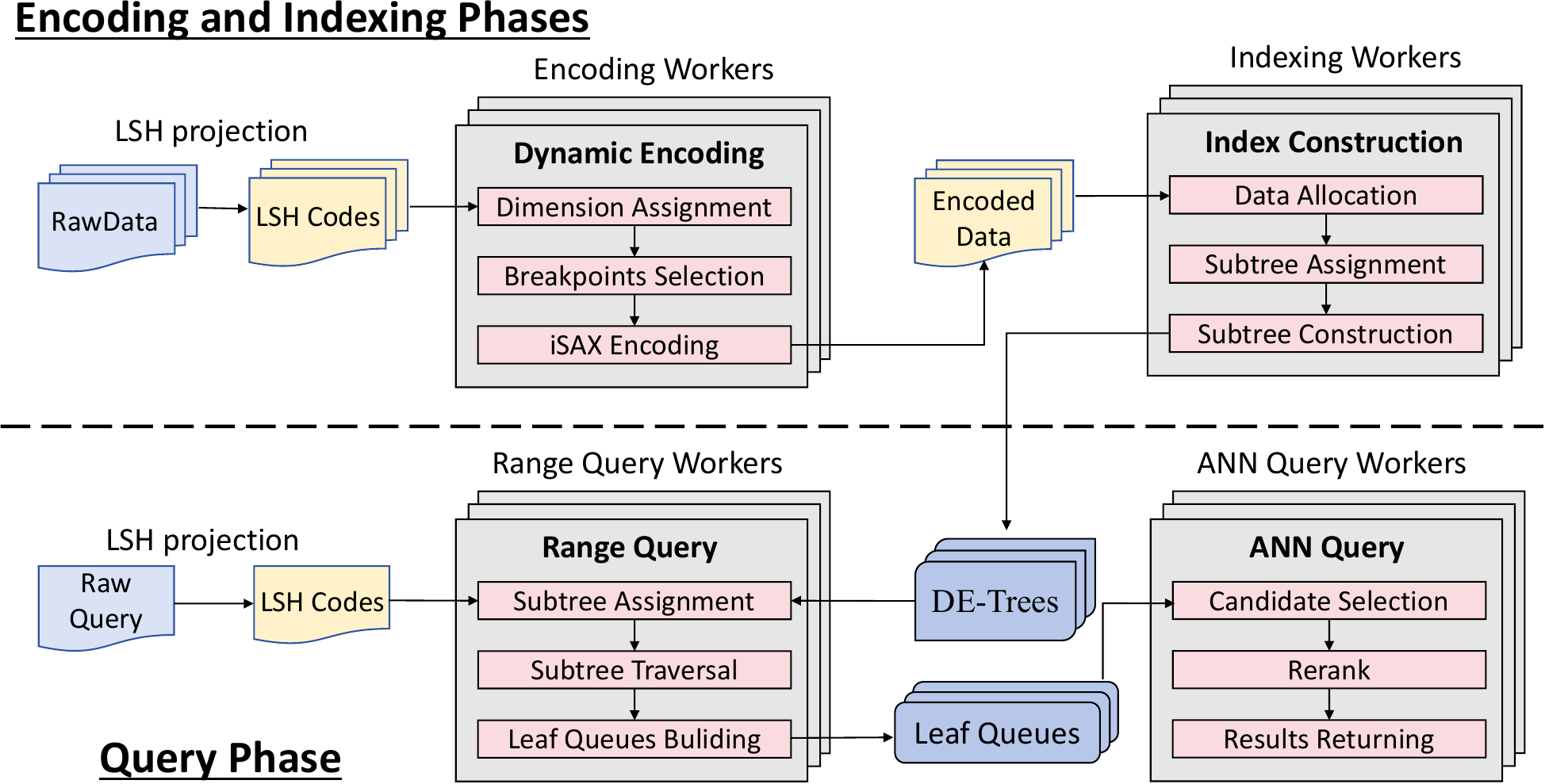}
	\caption{Overview of the PDET-LSH workflow.}
	\label{pdetlsh_overview}
\end{figure}

\section{The PDET-LSH Method} \label{pdetlsh}

In this section, we describe PDET-LSH, an in-memory LSH-based method that takes advantage of the parallelization opportunities provided by multicore CPUs. 
PDET-LSH has the same query accuracy as DET-LSH, while greatly improving its index-building and query-answering efficiency.

The parallelism strategies we design in PDET-LSH are governed by two main principles: eliminate synchronization overheads as much as possible, 
and balance the workload of multiple workers.
To achieve these principles, we design task-specific workload assignment strategies to balance workload across workers, and introduce a multi-queue buffering mechanism to store leaf nodes retrieved during range queries, thereby minimizing synchronization and waiting overheads among concurrent workers.
Note that our key design decision stems from the fact that in our proposed DE-Tree, the root node has a large number of child nodes. 
This allows us to introduce a lot of concurrency at a low synchronization cost.

Figure~\ref{pdetlsh_overview} provides an overview of the workflow for PDET-LSH.
In the encoding phase (dimension-partitioned) and indexing phase (data-partitioned), each worker independently handles its tasks on assigned subspace or data subset, performing breakpoints selection, iSAX encoding, and subtree construction.
During queries, range query workers traverse DE-Trees to build candidate leaf queues, while ANN query workers select and rerank candidates in parallel. 
This parallel design exploits both data-level and space-level independence, significantly improving indexing and query efficiency.

\begin{algorithm}[t]
\small
	\caption{Parallel Dynamic Encoding}                                 
	\label{parallel_dynamic_encoding}
	\LinesNumbered
	\KwIn{Parameters $K$, $L$, $n$, all points in projected spaces $P$, number of regions in each projected space $N_r$, number of workers $N_w$}
	\KwOut{A set of encoded points $EP$}
	Initialize the breakpoints set $B$ with size $L \cdot K \cdot (N_r+1)$; \\
        Initialize the encoded points set $EP$ with size $n \cdot L \cdot K$; \\
        \For{$wid=0$ to $N_w-1$}{
        $start\_dimension \leftarrow wid \cdot \lfloor \frac{K}{N_w} \rfloor$; \\
        $end\_dimension \leftarrow (wid+1) \cdot \lfloor \frac{K}{N_w} \rfloor$; \\
        Obtain the dimensions of all points in the scope ($start\_dimension$, $end\_dimension$) from $P$; \\
        Create a worker to execute Breakpoints Selection in the scope; \\
        Create a worker to execute Dynamic Encoding in the scope; \\
        }
	Wait for all workers to finish their work; \\
	\Return $EP$; \\
\end{algorithm}

\begin{algorithm}[t]
\small

	\caption{Parallel Create Index}                                                                       
	\label{parallel_create_index}
	\LinesNumbered
	\KwIn{Parameters $K$, $L$, $n$, encoded points set $EP$, number of workers $N_w$}
	\KwOut{A set of DE-Trees: $DETs=[T_1,...,T_L]$}
	\For{$i=1$ to $L$}{
		Initialize $T_i$ and generate $2^{K}$ first-layer nodes as the original leaf nodes; \\
		\For{$wid=0$ to $N_w - 1$}{
        $start\_point \leftarrow wid \cdot \lfloor \frac{n}{N_w} \rfloor$; \\
        $end\_point \leftarrow (wid+1) \cdot \lfloor \frac{n}{N_w} \rfloor$; \\
         Obtain the encoded points in the working scope ($start\_point$, $end\_point$) from $EP$; \\
			Create a worker to insert the encoded nodes that belong to the working scope into the appropriate first-layer tree nodes; \\
            Set a barrier to synchronize all workers; \\
            Create a worker to continuously and exclusively acquires first-level tree nodes and constructs subtrees rooted at those nodes; \\
		}
	}
	\Return $DETs$; \\
\end{algorithm}

\subsection{Encoding Phase} \label{encoding pdet}

Algorithm~\ref{parallel_dynamic_encoding} presents our design for parallel dynamic encoding. 
The main idea is that multiple workers can work independently within the scope of work assigned to them. 
Specifically, we first divide the work scope for $N_w$ workers, that is, specific dimensions within the $K$-dimensional projected space (lines 3-5). 
Then, each worker needs to create a worker (thread) and execute the breakpoints selection and dynamic encoding successively in all $L$ DE-Trees (lines 6-8). 
The decision of selecting breakpoints and encoding in each dimension is similar to that in Algorithm~\ref{dynamic_encoding}. 
When all workers finish their work, we get the encoded points set $EP$ (lines 9-10).

\subsection{Indexing Phase} \label{indexing pdet}

The root node of a DE-Tree may have up to $2^K$ child nodes (i.e., first-layer nodes). Each of these nodes acts as the root of a complete binary tree that recursively partitions the corresponding subspace. 
Each binary tree can be constructed independently if all data points are assigned to appropriate first-layer nodes. 
Therefore, the DE-Tree structure naturally lends itself to parallel construction.

Algorithm~\ref{parallel_create_index} presents the pseudocode for the parallel index creation. 
We initialize and build each DE-Tree sequentially (lines 1-2). 
First, we initialize $N_w$ workers for each DE-Tree and assign them working scopes (specific points in the dataset) (lines 3-6). 
Then, each worker creates a thread to insert the encoded points within the working scope into the appropriate first-layer nodes of the DE-Tree (line 7).
Up until all workers complete the task (line 8),
they continue to fetch unvisited first-layer nodes, and build a complete binary subtree with the corresponding first-layer node acting as the root node for that subtree (line 9).

\begin{figure}[tb] 
	\centering
	\includegraphics[width=0.85\linewidth]{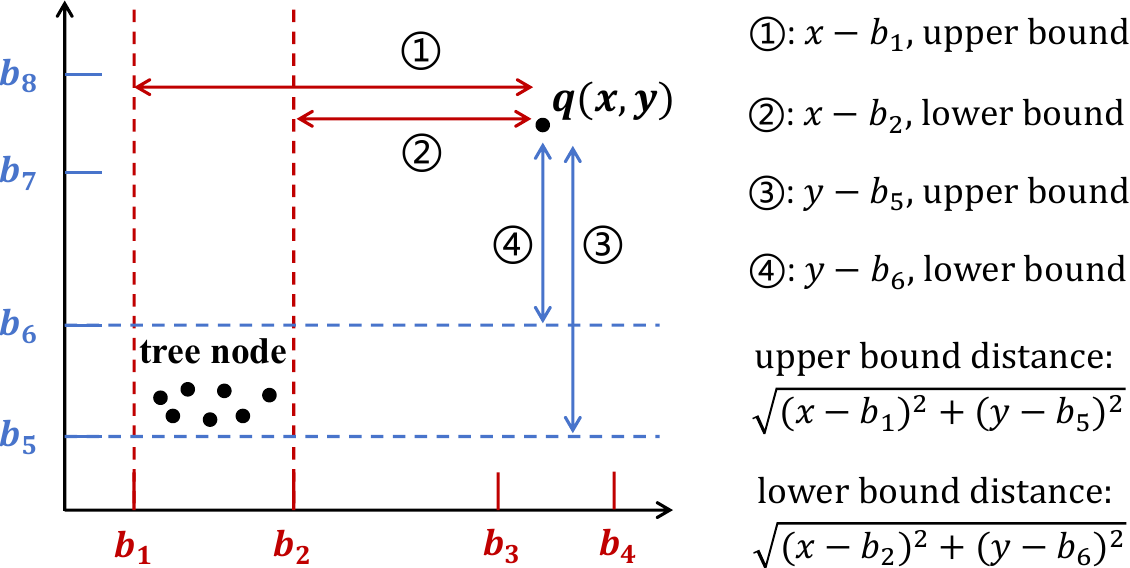}
	\caption{An example of calculating the upper and lower bound distances between a query $q$ and a DE-Tree node.}
	\label{dist}
\end{figure}

\subsection{Querying Phase} \label{query pdet}

Since the query strategy of PDET-LSH is based on the Euclidean distance metric, range queries can improve the efficiency of obtaining candidate points. 
In a DE-Tree, each space is divided into different regions by multiple breakpoints. 
The breakpoints on all sides of a region can be used to calculate the upper and lower bound distances between two points or between a point and a tree node. 

Figure~\ref{dist} gives an example for calculating the upper and lower bound distances between a query $q$ and a DE-Tree node in a two-dimensional space. 
Assume that the node is divided by breakpoints $b_1, b_2, b_5, b_6$, and the coordinates of $q$ are $(x,y)$. 
First, we need to calculate the upper and lower bound distances of each dimension based on the coordinates of $q$ and the breakpoints dividing the node, i.e., $x-b_1$ (upper bound), $x-b_2$ (lower bound), and $y-b_5$ (upper bound), $y-b_6$ (lower bound). 
Then, we use all the upper/lower bound distances obtained in each dimension to calculate the overall upper/lower bound distance: $\sqrt{(x-b_1)^2+(y-b_5)^2}$ and  $\sqrt{(x-b_2)^2+(y-b_6)^2}$. 
Therefore, the upper and lower bound distances between a query and a DE-Tree node can be calculated easily, which is suitable for range queries in PDET-LSH.

Range queries can effectively prune the search space using bound distances, improving query scalability and performance.
Algorithm~\ref{parallel_range_query} shows the parallel range query strategy we designed. 
First, we initialize $N_q$ queues to hold leaf nodes (lines 1-2), avoiding multiple workers competing for a single queue. 
Then, we create $N_w$ workers to continuously and exclusively traverse the subtrees and enqueue the leaf nodes that satisfy the distance bound conditions (lines 4-5).
Finally, merge all leaf nodes in $N_q$ queues into $S_{leaf}$ and return it (lines 6-7). 

For $c^2$-$k$-ANN query, the parallelization optimizations in PDET-LSH do not alter the DET-LSH query pipeline of DET-LSH (Algorithm \ref{ckann}); they only parallelize three of its steps, i.e., range query (line 5 in Algorithm \ref{ckann}), candidate selection (line 6 in Algorithm \ref{ckann}), and rerank (lines 8, 10 in Algorithm \ref{ckann}).
In fact, the parallel optimizations do not alter the outcomes of these three steps—namely, the set of leaf nodes, the candidate set, and the rerank results—which remain identical to those of Algorithm~\ref{ckann}.

We note that the results returned by PDET-LSH are the same as those returned by the optimized Algorithm~\ref{ckann} (in Section~\ref{queryoptimize}). 
That is, the query accuracy of PDET-LSH is the same as that of DET-LSH. 
In Section~\ref{comparison}, the experimental evaluation confirms that PEDT-LSH and DET-LSH have the same recall and overall ratio, outperforming other methods.

\begin{algorithm}[tb]
\small

	\caption{Parallel Range Query}                                                                     
	\label{parallel_range_query}
	\LinesNumbered
	\KwIn{A projected query point $q^\prime$, the search radius $r^\prime$, the index DE-Tree $T$, number of workers $N_w$, number of queues $N_q$}
	\KwOut{A set of leaf nodes $S_{leaf}$}
        \For{$i=0$ to $N_q - 1$}{
        Initialize the $i$-th queue $queues[i]$ to hold leaf nodes; \\
        }
	Initialize a leaf nodes set $S_{leaf} \leftarrow \varnothing$; \\
	\For{$i=0$ to $N_w - 1$}{
		Create a worker to continuously and exclusively acquires first-level tree nodes of $T$, performs Range Query with $q^\prime$ and $r^\prime$ on the subtree, inserts the qualified leaf nodes into a random queue; \\
        
	}
        Merge all leaf nodes in $N_q$ queues into $S_{leaf}$; \\
	\Return $S_{leaf}$;
\end{algorithm}

\section{Theoretical Analysis} \label{chapter5}

\subsection{Quality Guarantee}  \label{chapter5.1}

Let $\mathcal H _i(o)=[h_{i1}(o),...,h_{iK}(o)]$ denote a data point $o$ in the $i$-th projected space, where $i=1,...,L$. We define three events as follows:

\begin{itemize}
	\item \textbf{E1:} If there exists a point $o$ satisfying $\left\|o,q\right\| \leq r$, then its projected distance to $q$, i.e., $\left\|\mathcal H _i(o),\mathcal H _i(q)\right\|$, is smaller than $\epsilon r$ for some $i=1,...,L$;
	\item \textbf{E2:} If there exists a point $o$ satisfying $\left\|o,q\right\| > cr$, then its projected distance to $q$, i.e., $\left\|\mathcal H _i(o),\mathcal H _i(q)\right\|$, is smaller than $\epsilon r$ for some $i=1,...,L$;
	\item \textbf{E3:} Fewer than $\beta n$ points satisfying \textbf{E2} in dataset $\mathcal D$.
\end{itemize}

\begin{lemma}\label{lemma3}
	Given $K$ and $c$, setting $L=-\frac{1}{\ln{\alpha_1}}$ and $\beta=2-2\alpha_2^{-\frac{1}{ln\alpha_1}}$ such that $\alpha_1$, $\alpha_2$, and $\epsilon$ satisfy Equation \ref{eq2}, the probability that \textbf{E1} occurs is at least $1-\frac{1}{\mathrm{e}}$ and the probability that \textbf{E3} occurs is at least $\frac{1}{2}$.
	
	\begin{equation}  \label{eq2}
		\epsilon^2=\chi^2_{\alpha_1}(K)=c^2 \cdot \chi^2_{\alpha_2}(K).
	\end{equation}
\end{lemma}

\begin{proof}
	Given a point $o$ satisfying $\left\|o,q\right\| \leq r$, let $s=\left\|o,q\right\|$ and $s_i^\prime=\left\|\mathcal H _i(o),\mathcal H _i(q)\right\|$ denote the distances between $o$ and $q$ in the original space and in the $i$-th projected space, where $i=1,...,L$. From Equation \ref{eq2}, we have $\sqrt{\chi^2_{\alpha_1}(K)}=\epsilon$. For each independent projected space, from Lemma \ref{lemma2}, we have $\Pr{[s_i^\prime>s\sqrt{\chi^2_{\alpha_1}(K)}]} = \Pr{[s_i^\prime>\epsilon s]} = \alpha_1$. Since $s \leq r$, $\Pr{[s_i^\prime>\epsilon r]} \leq \alpha_1$. Considering $L$ projected spaces, we have $\Pr{[\textbf{E1}]}\geq1-\alpha_1^L=1-\frac{1}{\mathrm{e}}$. Likewise, given a point $o$ satisfying $\left\|o,q\right\| > cr$, let $s=\left\|o,q\right\|$ and $s_i^\prime=\left\|\mathcal H _i(o),\mathcal H _i(q)\right\|$ denote the distances between $o$ and $q$ in the original space and in the $i$-th projected space, where $i=1,...,L$. From Equation \ref{eq2}, we have $\sqrt{\chi^2_{\alpha_2}(K)}=\frac{\epsilon}{c}$. For each independent projected space, from Lemma \ref{lemma2}, we have $\Pr{[s_i^\prime>s\sqrt{\chi^2_{\alpha_2}(K)}]} = \Pr{[s_i^\prime>\frac{\epsilon s}{c}]} = \alpha_2$. Since $s > cr$, i.e., $\frac{s}{c} > r$, $\Pr{[s_i^\prime>\epsilon r]} > \alpha_2$. Considering $L$ projected spaces, we have $\Pr{[\textbf{E2}]} \leq 1-\alpha_2^L$, thus the expected number of such points in dataset $\mathcal D$ is upper bounded by $(1-\alpha_2^L) \cdot n$. By \textit{Markov's inequality}, we have $\Pr{[\textbf{E3}]}>1-\frac{(1-\alpha_2^L) \cdot n}{\beta n}=\frac{1}{2}.$
\end{proof}

\begin{theorem}\label{theorem1}
	DET-LSH (Algorithm~\ref{rcann}) answers an ($r$,$c$)-ANN query with at least a constant probability of $\frac{1}{2}-\frac{1}{\mathrm{e}}$.
\end{theorem}

\begin{proof}
	We show that when $\textbf{E1}$ and $\textbf{E3}$ hold at the same time, 
 Algorithm~\ref{rcann} returns an correct ($r$,$c$)-ANN result. 
 The probability of \textbf{E1} and \textbf{E3} occurring at the same time can be calculated as $\Pr{[\textbf{E1}\textbf{E3}]}=\Pr{[\textbf{E1}]}-\Pr{[\textbf{E1}\overline{\textbf{E3}}]} > \Pr{[\textbf{E1}]}-\Pr{[\overline{\textbf{E3}}]}=\frac{1}{2}-\frac{1}{\mathrm{e}}$. 
 When $\textbf{E1}$ and $\textbf{E3}$ hold at the same time, if Algorithm~\ref{rcann} terminates after getting at least $\beta n+1$ candidate points (line 7),
 due to $\textbf{E3}$, there are at most $\beta n$ points satisfying $\left\|o,q\right\| > cr$. 
 Thus we can get at least one point satisfying $\left\|o,q\right\| \leq cr$, and the returned point is obviously a correct result. 
 If the candidate set $S$ has no more than $\beta n+1$ points, but there exists at least one point in $S$ satisfying $\left\|o,q\right\| \leq cr$, 
 Algorithm~\ref{rcann} can also terminate and then return a result correctly (line 9). 
 Otherwise, it indicates that no points satisfying $\left\|o,q\right\| \leq cr$. 
 According to the Definition~\ref{def3} of ($r$,$c$)-ANN, nothing will be returned (line 10). 
 Therefore, when $\textbf{E1}$ and $\textbf{E3}$ hold at the same time, Algorithm~\ref{rcann} can always correctly answer an ($r$,$c$)-ANN query. 
 In other words, DET-LSH (Algorithm~\ref{rcann}) answers an ($r$,$c$)-ANN query with at least a constant probability of $\frac{1}{2}-\frac{1}{\mathrm{e}}$.
\end{proof}

\begin{theorem}\label{theorem2}
	DET-LSH (Algorithm~\ref{ckann}) answers a $c^2$-$k$-ANN query with at least a constant probability of $\frac{1}{2}-\frac{1}{\mathrm{e}}$.
\end{theorem}

\begin{proof}
	We show that when $\textbf{E1}$ and $\textbf{E3}$ hold at the same time, Algorithm~\ref{ckann} returns a correct $c^2$-$k$-ANN result. 
 Let $o_i^*$ be the $i$-th exact nearest point to $q$ in $\mathcal D$, we assume that $r_i^* = \left\|o_i^*,q\right\| > r_{min}$, where $r_{min}$ is the initial search radius and $i=1,...,k$. 
 We denote the number of points in the candidate set under search radius $r$ as $\lvert S_r \rvert$. 
 Obviously, when enlarging the search radius $r=r_{min}, r_{min} \cdot c, r_{min} \cdot c^2,...$, there must exist a radius $r_0$ satisfying $\lvert S_{r_0} \rvert < \beta n+k$ and $\lvert S_{c \cdot r_0} \rvert \geq \beta n+k$. 
 The distribution of $r_i^*$ has three cases:
	\begin{enumerate}
		\item \textbf{Case 1:} If for all $i=1,...,k$ satisfying $r_i^* \leq r_0$, which indicates the range queries in all $L$ index trees have been executed at $r=r_0$ (lines 3-8). Due to $\textbf{E1}$, all $r_i^*$ must in $S_{r_0}$. Since $S_{r_0} \subsetneqq S_{c \cdot r_0}$, all $r_i^*$ also must in $S_{c \cdot r_0}$. Therefore, Algorithm~\ref{ckann} returns the exact $k$ nearest points $o_i^*$ to $q$.
		\item \textbf{Case 2:} If for all $i=1,...,k$ satisfying $r_i^* > r_0$, all $r_i^*$ not belong to $S_{r_0}$. Since Algorithm~\ref{ckann} may terminate after executing range queries in part of $L$ index trees at $r=c \cdot r_0$ (line 8), we cannot guarantee that $r_i^* \leq c \cdot r_0$. However, due to $\textbf{E3}$, there are at least $k$ points $o_i$ in $S_{c \cdot r_0}$ satisfying $\left\|o_i,q\right\| \leq c^2r_0$, $i=1,...,k$. Therefore, we have $\left\|o_i,q\right\| \leq c^2r_0 \leq c^2r_i^*$, i.e., each $o_i$ is a $c^2$-ANN point for corresponding $o_i^*$.
		\item \textbf{Case 3:} If there exists an integer $m \in (1,k)$  such that for all $i=1,...,m$ satisfying $r_i^* \leq r_0$ and for all $i=m+1,...,k$ satisfying $r_i^* > r_0$, indicating that \textbf{Case 3} is a combination of \textbf{Case 1} and \textbf{Case 2}. For each $i \in [1,m]$, Algorithm~\ref{ckann} returns the exact nearest point $o_i^*$ to $q$ based on \textbf{Case 1}. For each $i \in [m+1,k]$, Algorithm~\ref{ckann} returns a $c^2$-ANN point for $o_i^*$ based on \textbf{Case 2}.
	\end{enumerate}
	Therefore, when $\textbf{E1}$ and $\textbf{E3}$ hold simultaneously, Algorithm~\ref{ckann} can always correctly answer a $c^2$-$k$-ANN query, i.e., DET-LSH (Algorithm~\ref{ckann}) returns a $c^2$-$k$-ANN with at least a constant probability of $\frac{1}{2}-\frac{1}{\mathrm{e}}$.
\end{proof}

\begin{theorem}\label{theorem3}
	PDET-LSH answers a $c^2$-$k$-ANN query with at least a constant probability of $\frac{1}{2}-\frac{1}{\mathrm{e}}$.
\end{theorem}
\begin{proof}
The parallelization optimization of PDET-LSH does not change the query pipeline of DET-LSH (Algorithm~\ref{ckann}) in steps where theoretical guarantees for result quality are required (e.g., termination conditions).
Specifically, PDET-LSH parallelizes three steps in the query pipeline: range query (line 5), candidate selection (line 8), and rerank (line 11).
In fact, the parallel optimizations do not alter the outcomes of these three steps—namely, the set of leaf nodes, the candidate set, and the rerank results—which remain identical to those of Algorithm~\ref{ckann}.
Moreover, as shown in the proof of Theorem~\ref{theorem1} and Theorem~\ref{theorem2}, the termination condition of algorithms is the key to theoretical guarantees because it can guarantee the quality of the returned results.
The termination condition of PDET-LSH is the same as that of Algorithm~\ref{ckann} (lines 7 and 9), so PDET-LSH and DET-LSH return the same results.
Figures~\ref{scalability} and~\ref{diffk} confirm that under identical parameter settings, PDET-LSH and DET-LSH achieve the same query accuracy (i.e., return identical results).
Therefore, PDET-LSH has the same quality guarantee as DET-LSH, i.e., PDET-LSH returns a $c^2$-$k$-ANN with at least a constant probability of $\frac{1}{2}-\frac{1}{\mathrm{e}}$.
\end{proof}

\subsection{Parameter Settings} \label{chapter5.2}

\subsubsection{DET-LSH} \label{para_det}

\begin{figure}[tb] 
	\centering
	\includegraphics[width=0.53\linewidth]{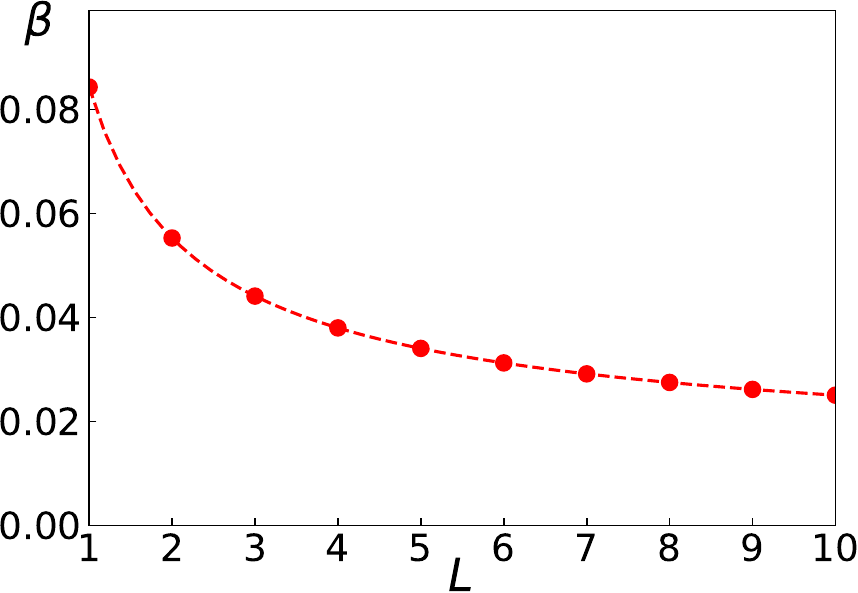}
	\caption{Illustration of theoretical $\beta$ when $L$ varies (for $K=16$ and $c=1.5$), which is in line with Lemma \ref{lemma3}.}
	\label{betal}
\end{figure}

The performance of DET-LSH is affected by several parameters: $L$, $K$, $\beta$, $c$, and so on. 
According to Lemma~\ref{lemma3}, when $K$ and $c$ are set as constants, there is a mathematical relationship between $L$ and $\beta$. 
We set $K=16$ and $c=1.5$ by default, and Figure~\ref{betal} shows the theoretical $\beta$ as $L$ changes, which is in line with Lemma~\ref{lemma3}. Figure~\ref{betal} illustrates that $\beta$ and $L$ have a negative correlation. 
Theoretically, a greater $\beta$ means a higher fault tolerance when querying, so the accuracy of DET-LSH is improved. 
Meanwhile, a greater $L$ means fewer correct results are missed when querying, so the accuracy of DET-LSH can also be improved. 
However, both greater $\beta$ and greater $L$ will reduce query efficiency, so we need to find a balance between $\beta$ and $L$. 
As shown in Figure~\ref{betal}, $L=4$ is a good choice because as $L$ increases, $\beta$ drops rapidly until $L=4$, and then $\beta$ drops slowly. 
Therefore, we choose $L=4$ as the default value. 

For the initial search radius $r_{min}$, we follow the selection scheme proposed in~\cite{pmlsh}. Specifically, to reduce the number of iterations for different $r$ and terminate the query process faster, 
we find a \enquote{magic} $r_{min}$ that satisfies the following conditions: 1) when $r=r_{min}$ in Algorithm~\ref{ckann}, the number of candidate points in $S$ satisfies $\lvert S \rvert \geq \beta n+k$; 
2) when $r=\frac{r_{min}}{c}$ in Algorithm~\ref{ckann}, the number of candidate points in $S$ satisfies $\lvert S \rvert < \beta n+k$. 
Since DET-LSH can implement dynamic incremental queries as $r$ increases, the choice of $r_{min}$ is expected to have a relatively small impact on its performance.

\subsubsection{PDET-LSH} \label{para_pdet}

Even though PDET-LSH improves the indexing and query efficiency of DET-LSH, it has the same query accuracy as DET-LSH, that is, it returns the same results.
Therefore, we choose for PDET-LSH the same parameters as for DET-LSH; the choice of these parameters is discussed in Section~\ref{para_det}.

\subsection{Complexity Analysis}  \label{chapter5.3}

\subsubsection{DET-LSH}

In the encoding and indexing phases, DET-LSH has time cost $\mathcal{O}(n(d+\log N_r))$, and space cost $\mathcal{O}(n)$.
The time cost comes from four parts: (1) computing hash values for $n$ points, $\mathcal{O}(L \cdot K \cdot n \cdot d)$; 
(2) breakpoints selection, $\mathcal{O}(L \cdot K \cdot n \cdot \log N_r)$; 
(3) dynamic encoding, $\mathcal{O}(L \cdot K \cdot n \cdot \log N_r)$; 
and (4) constructing $L$ DE-Trees, $\mathcal{O}(L \cdot n \cdot K \cdot \log N_r)$. 
Since both $K=\mathcal{O}(1)$ and $L=\mathcal{O}(1)$ are constants, the total time cost is $\mathcal{O}(n(d+\log N_r))$. 
Obviously, the size of encoded points and $L$ DE-Trees are both $\mathcal{O}(L \cdot K \cdot n)=\mathcal{O}(n)$.

In the query phase, DET-LSH has time cost $\mathcal{O}(n(\beta d+\log N_r))$. 
The time cost comes from four parts: 
(1) computing hash values for the query point $q$, $\mathcal{O}(L \cdot K \cdot d)=\mathcal{O}(d)$; 
(2) finding candidate points in $L$ DE-Trees, $\mathcal{O}(L \cdot K^2 \cdot \log N_r \cdot \frac{n}{max\_size} + L \cdot K \cdot n)=\mathcal{O}(n\log N_r)$; 
(3) computing the real distance of each candidate point to $q$, $\mathcal{O}(\beta nd)$ cost; 
and (4) finding the $top$-$k$ points to $q$, $\mathcal{O}(\beta n \log k)$. 
The total time cost in the query phase is $\mathcal{O}(n(\beta d+\log N_r))$.

\subsubsection{PDET-LSH}

In the encoding and indexing phases, PDET-LSH has time cost $\mathcal{O}(\frac{n(d+\log N_r)}{N_w})$, and space cost $\mathcal{O}(nN_w\log^2 N_r)$.
The time cost derives from three components: (1) computing hash values for $n$ points in parallel, with complexity $\mathcal{O}(\frac{L \cdot K \cdot n \cdot d}{N_w})$; 
(2) dynamic encoding in parallel, with complexity $\mathcal{O}(\frac{L \cdot K \cdot n \cdot \log N_r}{N_w})$; 
and (3) constructing $L$ DE-Trees in parallel, with complexity $\mathcal{O}(\frac{L \cdot n \cdot K \cdot \log N_r}{N_w})$. 
Therefore, the total time cost is $\mathcal{O}(\frac{n(d+\log N_r)}{N_w})$. 

The space cost relates to two components:
(1) the encoded points, with complexity $\mathcal{O}(L \cdot K \cdot n)=\mathcal{O}(n)$, and
(2) the $L$ DE-Trees built by $N_w$ workers, with complexity $\mathcal{O}(N_w \cdot L \cdot K \cdot \frac{n}{max\_size} \cdot \log^2 N_r)=\mathcal{O}(nN_w\log^2 N_r)$. 
Thus, the total space cost is $\mathcal{O}(nN_w\log^2 N_r)$.

In the query phase, PDET-LSH has $\mathcal{O}(\frac{n}{N_w}(\beta d+\log N_r))$ time cost, which derives from four components: 
(1) computing hash values for the query point $q$, with complexity $\mathcal{O}(\frac{L \cdot K \cdot d}{N_w})=\mathcal{O}(\frac{d}{N_w})$; 
(2) finding candidate points in the $L$ DE-Trees, with complexity $\mathcal{O}(L \cdot K^2 \cdot \log N_r \cdot \frac{n}{max\_size} \cdot \frac{1}{N_w} + L \cdot K \cdot n \cdot \frac{1}{N_w})=\mathcal{O}(\frac{n\log N_r}{N_w})$; 
(3) computing the real distance of each candidate point to $q$, with complexity $\mathcal{O}(\frac{\beta nd}{N_w})$ cost; 
and (4) finding the $top$-$k$ points to $q$, with complexity $\mathcal{O}(\beta n \log k)$. 
The total time cost of the query phase is thus $\mathcal{O}(\frac{n}{N_w}(\beta d+\log N_r))$.

\begin{table}
\vspace*{-0.3cm}
	\centering
	\caption{Datasets}
	\label{table2}
    	\begin{tabular}{cccc}
		\toprule
		\textbf{Dataset} & \textbf{Cardinality} & \textbf{Dim.} & \textbf{Type} \\
		\midrule
		Msong & 994,185 & 420 & Audio\\
		Deep1M & 1,000,000 & 256 & Image\\
		Sift10M & 10,000,000 & 128 & Image\\
		TinyImages80M & 79.302,017 & 384 & Image\\
		Sift100M & 100,000,000 & 128 & Image\\
		Yandex Deep500M & 500,000,000 & 96 & Image\\
		Microsoft SPACEV500M  & 500,000,000 & 100 & Text\\
		Microsoft Turing-ANNS500M & 500,000,000 & 100 & Text\\
		\bottomrule
	\end{tabular}
\end{table}

\section{Experimental Evaluation} \label{chapter6}

In this section, we evaluate the performance of DET-LSH and PDET-LSH, conduct comparative experiments with the state-of-the-art LSH-based methods.
DET-LSH and PDET-LSH are implemented in C and C++ and compiled using the -O3 optimization.
All LSH-based methods except for PDET-LSH use a single thread. 
We acknowledge that there are no parallel implementations for competitors, but given the increasing of datasets size and importance of LSH-based methods, we describe a parallel solution, which could also instigate other LSH-based solutions to do work in this area.
PDET-LSH uses Pthreads and OpenMP for parallelization and SIMD for accelerating calculations.
All experiments are conducted on a machine with 2 AMD EPYC 9554 CPUs @ 3.10GHz and 756 GB RAM, running on Ubuntu v22.04.

\subsection{Experimental Setup}

\noindent \textbf{Datasets and Queries.} We select eight high-dimensional commonly used datasets for ANN search. 
Table \ref{table2} shows the key statistics of the datasets. Note that points in \textit{Sift10M} and \textit{Sift100M} are randomly chosen from the \textit{Sift1B} dataset\footnote{http://corpus-texmex.irisa.fr/}. 
Similarly, \textit{Yandex Deep500M}, \textit{Microsoft SPACEV500M}, and \textit{Microsoft Turing-ANNS500M} are also randomly sampled from their 1B-scale datasets\footnote{https://big-ann-benchmarks.com/neurips21.html}, which are the largest datasets our server can handle.
We randomly select 100 data points as queries and remove them from the original datasets.

\noindent \textbf{Evaluation Measures.} We adopt five measures to evaluate the performance of all methods: index size, indexing time, query time, recall, and overall ratio.
For a query $q$, we denote the result set as $R=\{o_1,...,o_k\}$ and the exact $k$-NNs as $R^*=\{o_1^*,...,o_k^*\}$, \textit{recall} is defined as $\frac{\lvert R \cap R^* \rvert}{k}$ and \textit{overall ratio} is defined as $\frac{1}{k} \sum_{i=1}^{k} \frac{\left\|q,o_i\right\|}{\left\|q,o_i^*\right\|}$~\cite{dblsh}.
To evaluate the parallel performance of PDET-LSH, we adopt an additional measure: speedup ratio. 
\textit{Speedup ratio} is defined as $S_p=\frac{T_1}{T_p}$, where $p$ is the number of threads used, $T_1$ is the running time of non-parallelized method, and $T_p$ is the running time of parallel method with $p$ threads. 

\noindent \textbf{Benchmark Methods.} We compare DET-LSH and PDET-LSH with five state-of-the-art methods, including three LSH-based methods: DB-LSH~\cite{dblsh}, LCCS-LSH~\cite{lccslsh}, and PM-LSH~\cite{pmlsh},
and two non-LSH-based methods: HNSW~\cite{malkov2018efficient} (graph-based) and IMI-OPQ~\cite{ge2013optimized} (quantization-based).

\begin{figure}[tb]
\vspace*{-0.3cm}
	\centering
	\includegraphics[width=0.8\linewidth]{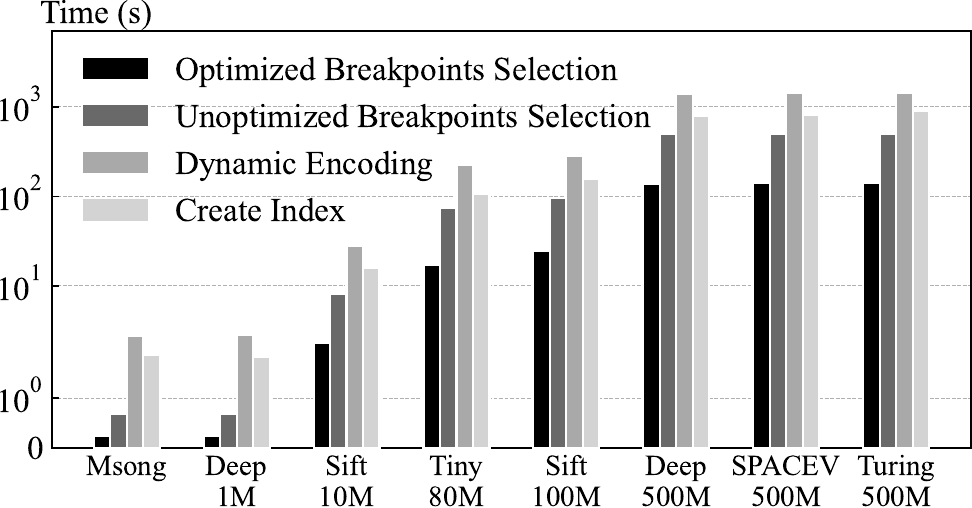}
	\caption{Running time break-down for the DET-LSH encoding and indexing phases.}
	\label{encodingandindexing}
\end{figure}

\begin{figure}[tb] 
	\centering
	\includegraphics[width=0.85\linewidth]{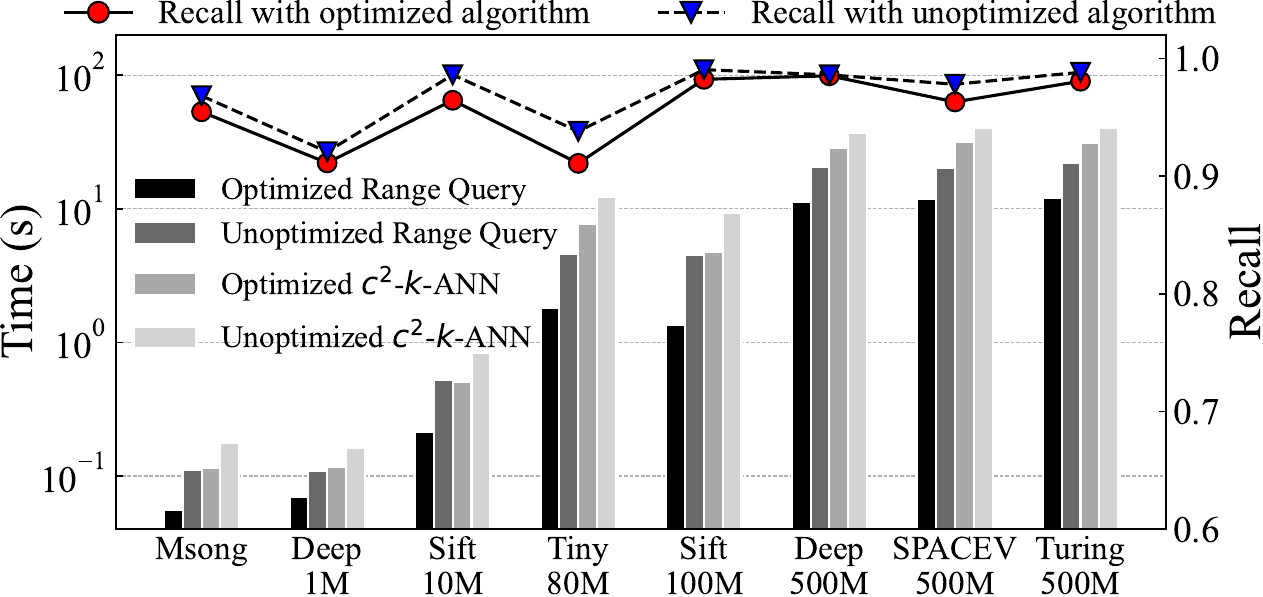}
	\caption{Running time and recall of optimized/non-optimized query-phase algorithms of DET-LSH.}
	\label{querytime}
\end{figure}

\begin{figure*}[tb] 
\vspace*{-0.3cm}
	\centering
	\includegraphics[width=0.85\linewidth]{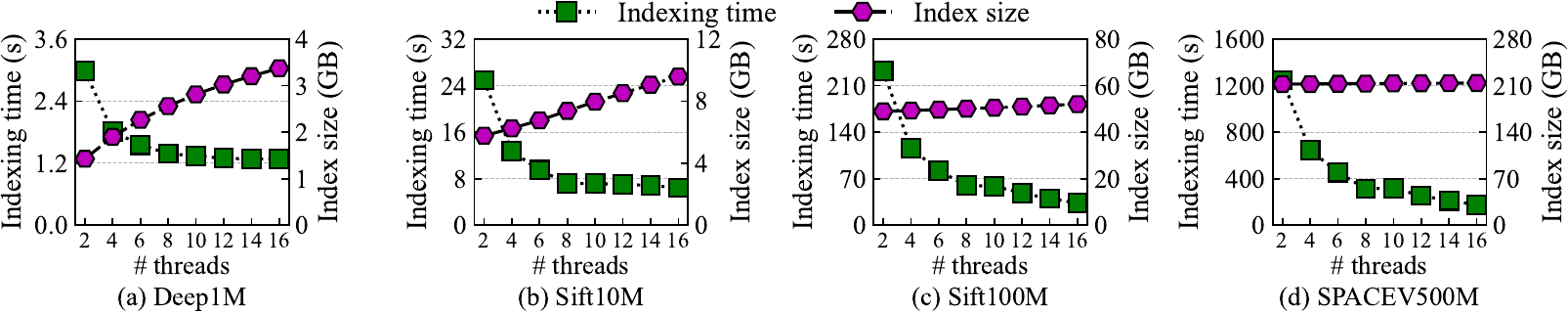}
	\caption{Indexing performance of PDET-LSH when varying the number of threads.}
	\label{diff_threads_indexing}
\end{figure*}

\begin{figure*}[tb] 
	\centering
	\includegraphics[width=0.85\linewidth]{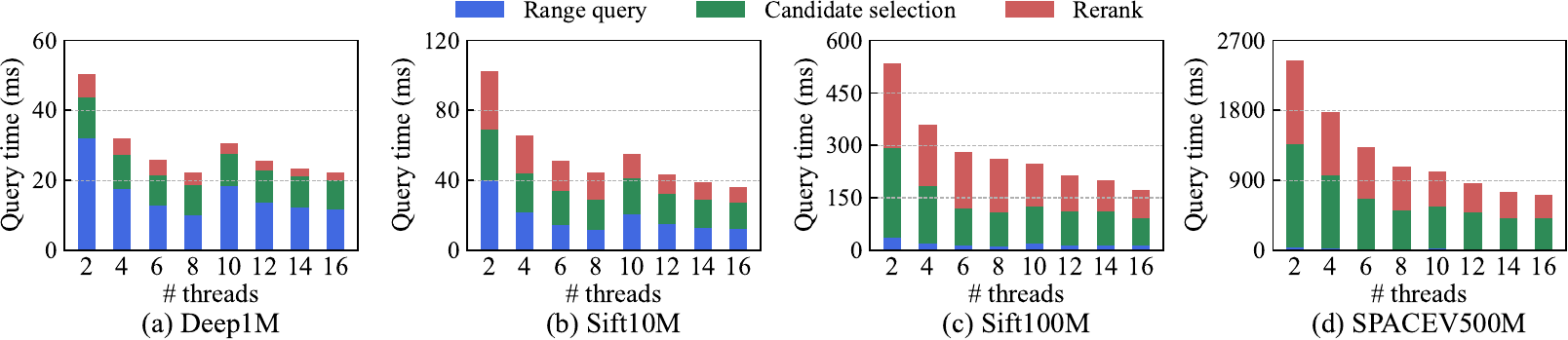}
	\caption{Query performance of PDET-LSH when varying the number of threads.}
	\label{diff_threads_query_bar}
\end{figure*}

\begin{figure}[tb] 
	\centering
	\includegraphics[width=0.85\linewidth]{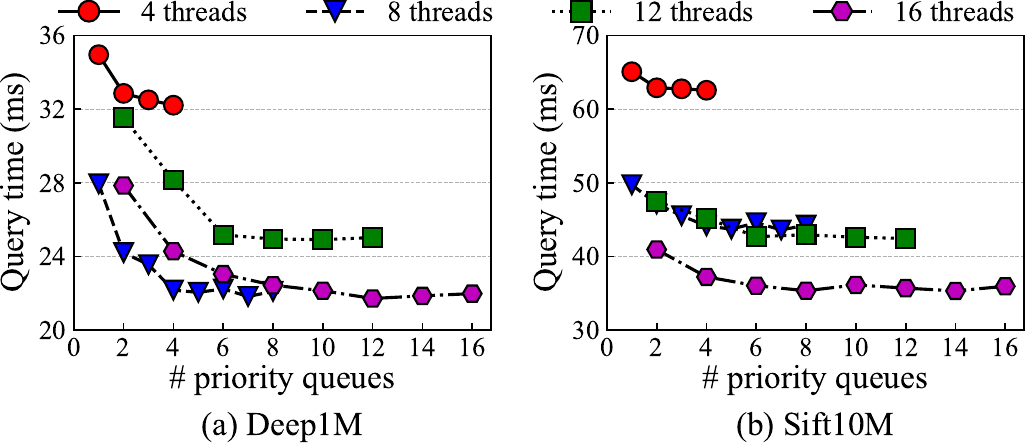}
	\caption{Query performance of PDET-LSH when varying the number of queues.}
	\label{diff_queue}
\end{figure}

\noindent \textbf{Parameter Settings.} $k$ in $k$-ANN is set to 50 by default.
For DET-LSH, the parameters are set as described in Section~\ref{para_det}.
For PDET-LSH, we set $N_w=8$ and $N_q=4$ (based on the experimental results in Section~\ref{pdet_query}).
For competitors, the parameter settings follow their source codes or papers. 
To make a fair comparison, we set $\beta=0.1$ and $c=1.5$ for DET-LSH, DB-LSH, PM-LSH. 
For DB-LSH, $L=5$, $K=12$, $w=4c^2$. 
For LCCS-LSH, $m=64$. For PM-LSH, $s=5$, $m=15$. 
For HNSW, $efConstruction=200$, $M=25$.
For IMI-OPQ, $M=2$, $K_{coarse}=2^{16}$, $K_{fine}=2^{8+16}$.

\subsection{Self-evaluation of DET-LSH} \label{selfevaluation}

\subsubsection{Encoding and Indexing Phase} \label{encodeoptimize}

Figure~\ref{encodingandindexing} details the algorithm runtimes for encoding and indexing. We found that dynamic encoding is slower than index creation, as locating a point's region per dimension among 256 options remains costly despite binary search optimization. 
Moreover, the optimized breakpoints selection (Section~\ref{Encoding Phase}), using QuickSelect with a divide-and-conquer strategy, reduces time complexity from $\mathcal{O}(n\log n)$ to $\mathcal{O}(n\log N_r)$, achieving a 3x speedup over the unoptimized version.

\subsubsection{Query Phase} \label{queryoptimize}

In practice, computing distances for all points in a leaf node is expensive when its upper bound distance to $q^\prime$ exceeds the search radius. 
Empirically, with a properly chosen leaf size in Algorithm~\ref{create_index}, most points in such leaf nodes lie within the search radius. 
Accordingly, we optimize subtree traversal in two ways: 
(1) we relax candidate selection by adding all points of a leaf node to $S$ whenever its lower bound distance to $q^\prime$ does not exceed $r$; 
(2) we maintain a priority queue of visited leaf nodes ordered by their lower bound distances to $q^\prime$, so that nodes closer to $q^\prime$ contribute candidates earlier. 
As shown in Figure~\ref{querytime}, with a modest loss in accuracy, the optimized Algorithm~\ref{range_query} and Algorithm~\ref{ckann} reduce query time by up to 50\% and 30\%, respectively.

\subsection{Self-evaluation of PDET-LSH} \label{selfevaluation_pdet}

\subsubsection{Encoding and Indexing Phase}

Since the encoding and indexing phases are tightly coupled, we report encoding performance as part of the indexing cost in subsequent experiments. 
Figure~\ref{diff_threads_indexing} presents the indexing performance of PDET-LSH with varying numbers of threads, while Figure~\ref{diff_threads_ratio}(a) further reports the corresponding speedup over DET-LSH. 
Overall, indexing time decreases substantially as the number of threads increases from 2 to 8, but exhibits diminishing returns from 8 to 16 due to intensified resource contention. 
Meanwhile, index size grows with the number of threads, especially on smaller datasets (1M and 10M), because per-thread index allocation is used to avoid contention; this growth becomes less pronounced on larger datasets (100M and 500M), where data-dependent index space dominates. 
Considering both indexing time and index size, using 8 threads achieves a favorable trade-off across all datasets.

\begin{figure}[tb] 
	\centering
	\includegraphics[width=0.85\linewidth]{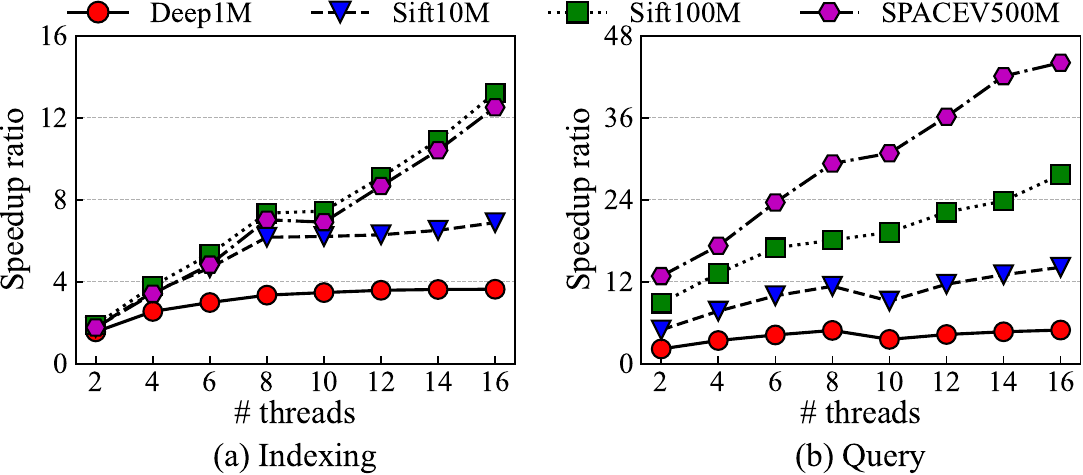}
	\caption{Parallel performance of PDET-LSH's indexing and query process when varying the number of threads.}
	\label{diff_threads_ratio}
\end{figure}

\begin{figure*}[tb]
\vspace*{-0.3cm}
	\centering
	\includegraphics[width=0.85\linewidth]{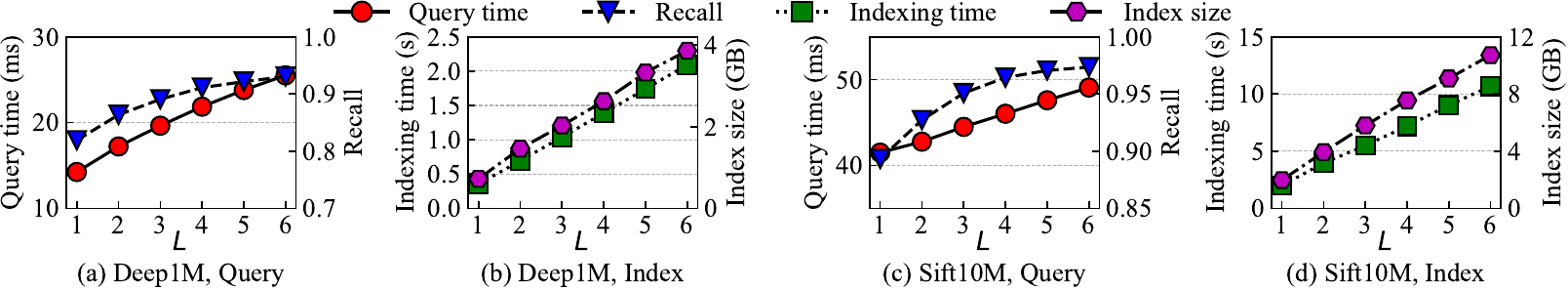}
	\caption{Performance of PDET-LSH under different number of projected spaces ($L$).}
	\label{diff_l}
\end{figure*}

\begin{figure*}[tb] 
	\centering
	\includegraphics[width=0.85\linewidth]{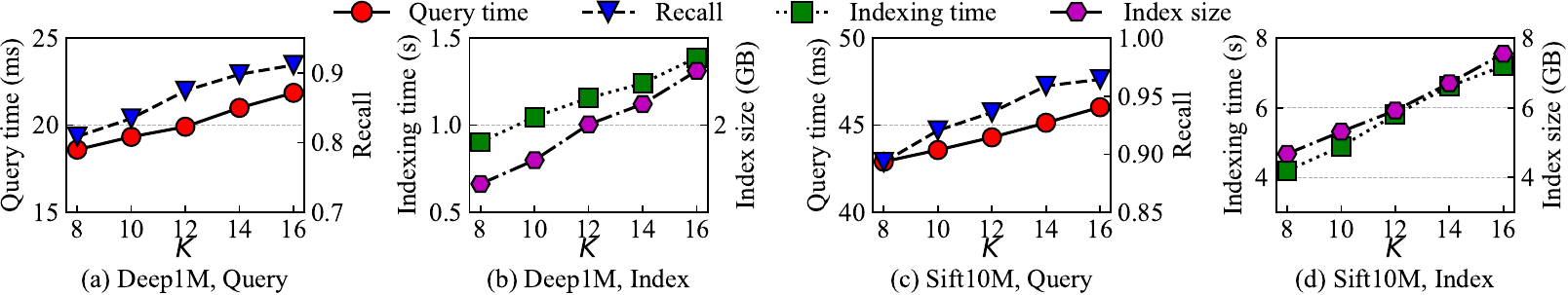}
	\caption{Performance of PDET-LSH under different dimension of projected spaces ($K$).}
	\label{diff_k}
\end{figure*}

\subsubsection{Query Phase} \label{pdet_query}

To study the effect of the number of queues $N_q$ on query performance, we evaluate PDET-LSH on the \textit{Deep1M} and \textit{Sift10M} datasets, as shown in Figure~\ref{diff_queue}. 
We observe that query efficiency achieves the best when $N_q=\frac{N_w}{2}$, i.e., the number of queues is half the number of threads. 
Increasing $N_q$ beyond this point yields negligible performance gains while introducing additional queue-management overhead. 
Accordingly, in all subsequent experiments, we fix $N_q=\frac{N_w}{2}$ regardless of the number of threads.

Since PDET-LSH produces identical results to DET-LSH, both methods achieve the same query accuracy; hence, accuracy results are omitted in the remainder of this section for brevity.
Figure~\ref{diff_threads_query_bar} illustrates the per-step query performance of $c^2$-$k$-ANN in PDET-LSH under different numbers of threads.
Figure~\ref{diff_threads_ratio}(b) reports the query speedup of PDET-LSH over DET-LSH.
We found that as dataset size grows, the relative cost of the range query diminishes, while candidate generation and $k$-NN refinement dominate due to the need to process $\beta n + k$ candidates. 
On \textit{Deep1M} and \textit{Sift10M}, query efficiency peaks at 8 threads, whereas larger datasets (\textit{Sift100M} and \textit{SPACEV500M}) benefit from additional threads.
This is because the AMD EPYC Processor we use has 8 cores per CCX (CPU Complex) and shares an L3 cache~\cite{amd}, where cross-CCX data synchronization needs to go through memory, so the data processing speed will decrease when the number of cores increases from 8 to 10. 
Moreover, PDET-LSH exhibits super-linear speedup ($S_p > p$) owing to parallel-aware query optimization and SIMD-accelerated distance computations, an effect that becomes more pronounced on large datasets. 
Overall, 8 threads provide a robust and effective configuration across datasets of different scales.

\begin{figure} [t!]
	\centering
	\subfigcapskip=5pt
	\subfigure[Datasets of different sizes.]{
		\includegraphics[width=0.46\linewidth]{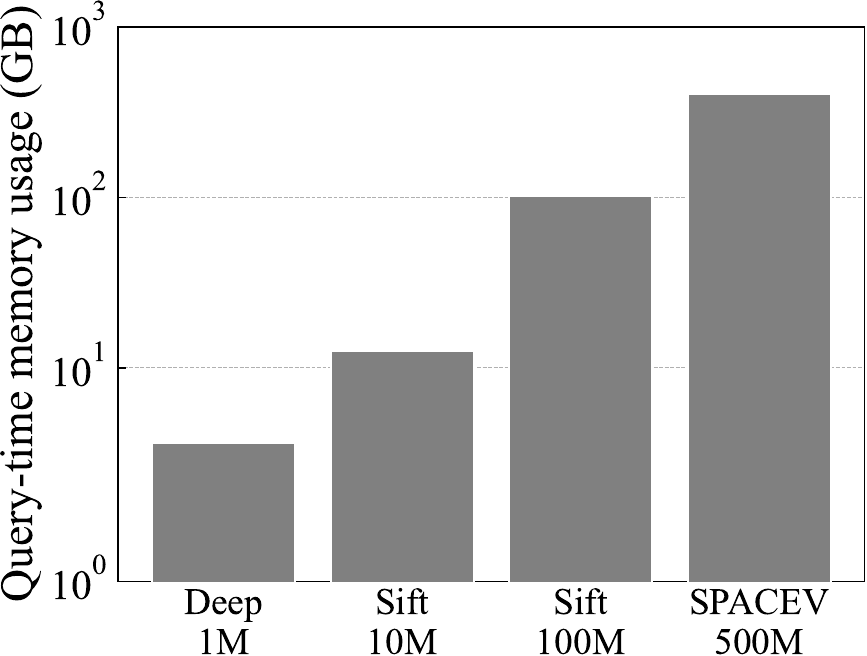}
		\label{query_memory}}
	\subfigure[Scalability on Sift dataset.]{
		\includegraphics[width=0.46\linewidth]{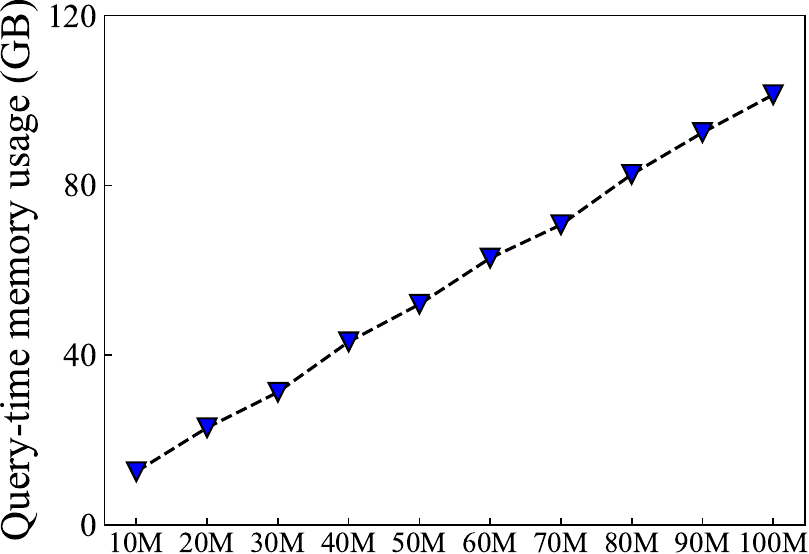}
		\label{query_memory_scalability}}
	\caption{Query-time memory overhead of PDET-LSH.}
	\label{query_memory_overhead}
\end{figure}

\begin{figure}[tb] 
	\centering
	\includegraphics[width=0.8\linewidth]{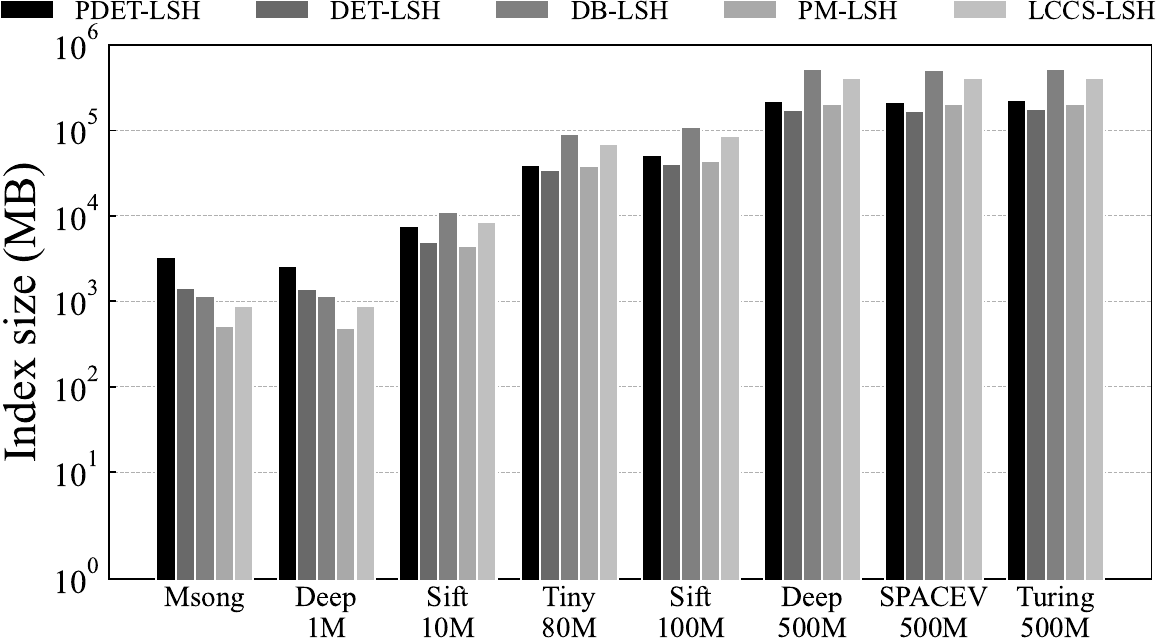}
	\caption{Index size.}
	\label{indexsize}
\end{figure}

\begin{figure}[tb] 
	\centering
	\includegraphics[width=0.8\linewidth]{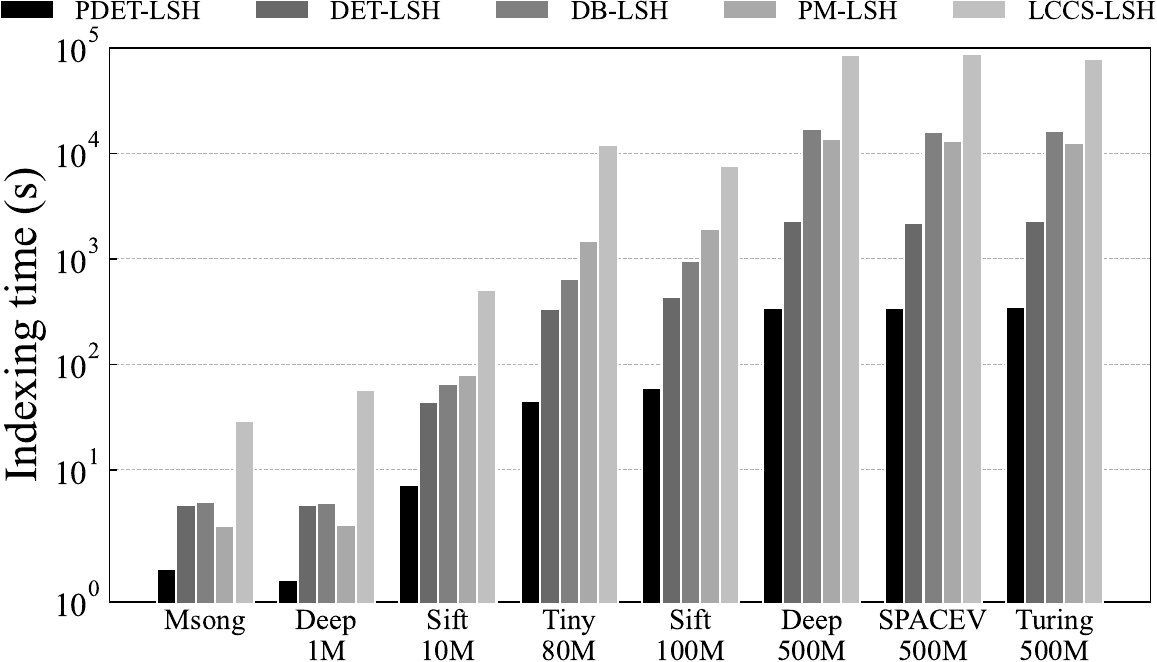}
	\caption{Indexing time.}
	\label{indexingtime}
\end{figure}

\subsubsection{Parameter study on $L$ and $K$} \label{para_study}
Figure~\ref{diff_l} and Figure~\ref{diff_k} show the performance of PDET-LSH under different $L$ and $K$, respectively.
We found that increasing $L$ and $K$ improves query accuracy, as it increases the accuracy of the LSH projections; however, this comes at the cost of reduced query efficiency, higher indexing overhead, and larger index size.
When $K=4$ and $L=16$, DET-LSH achieves a good trade-off between performance and overhead on multiple datasets.
In practice, users can dynamically choose appropriate parameter settings based on the specific application scenario and available resources.

\subsubsection{Query-time memory overhead} \label{memory_overhead}
Figure~\ref{query_memory_overhead} shows the query-time memory consumption of PDET-LSH, including the dataset size (as the dataset needs to be kept in memory for reranking). 
Specifically, Figure~\ref{query_memory} provides the memory usage of PDET-LSH on datasets of different sizes and dimensions during query processing.
Figure~\ref{query_memory_scalability} measures scalability, that is the query-time memory overhead of PDET-LSH when varying scales of the same dataset (Sift).
In practice, users can adaptively tune the parameters based on available memory resources. As illustrated in Figures~\ref{diff_l} and~\ref{diff_k}, larger $L$ or $K$ improve query accuracy at the expense of higher memory overhead, and vice versa.

\begin{figure*}[tb] 
\vspace*{-0.3cm}
	\centering
	\includegraphics[width=0.72\linewidth]{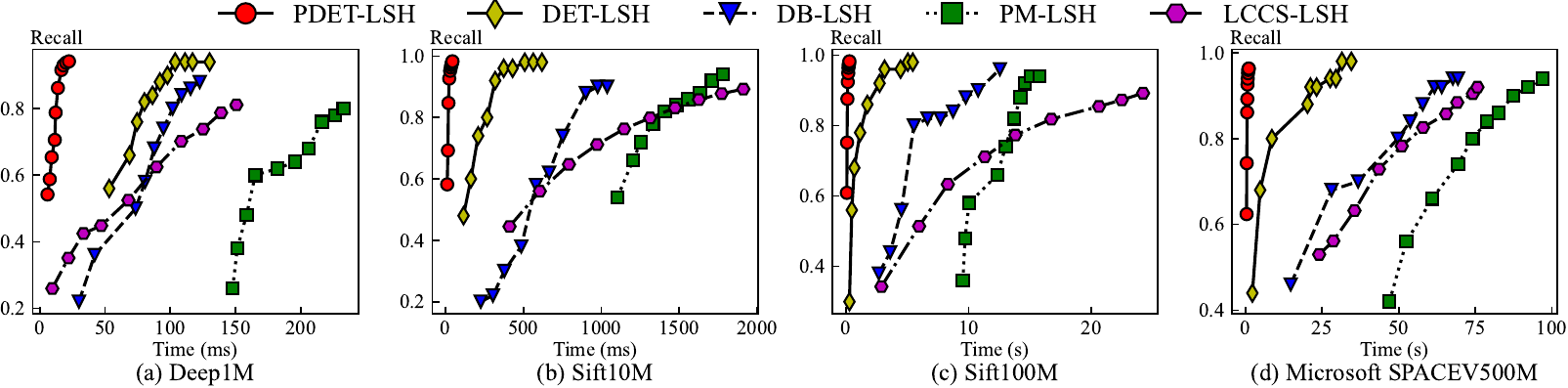}
    \vspace*{-0.1cm}
	\caption{Recall-time curves.}
	\label{recalltime}
\end{figure*}

\begin{figure*}[tb] 
	\centering
	\includegraphics[width=0.72\linewidth]{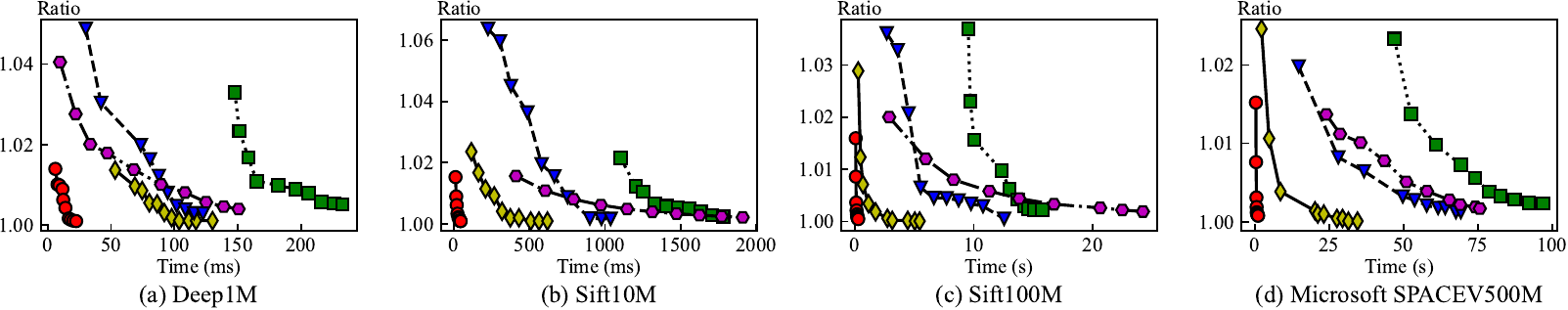}
    \vspace*{-0.1cm}
	\caption{Overall ratio-time curves.}
	\label{ratiotime}
\end{figure*}
    
\begin{figure*}[tb] 
	\centering
	\includegraphics[width=0.72\linewidth]{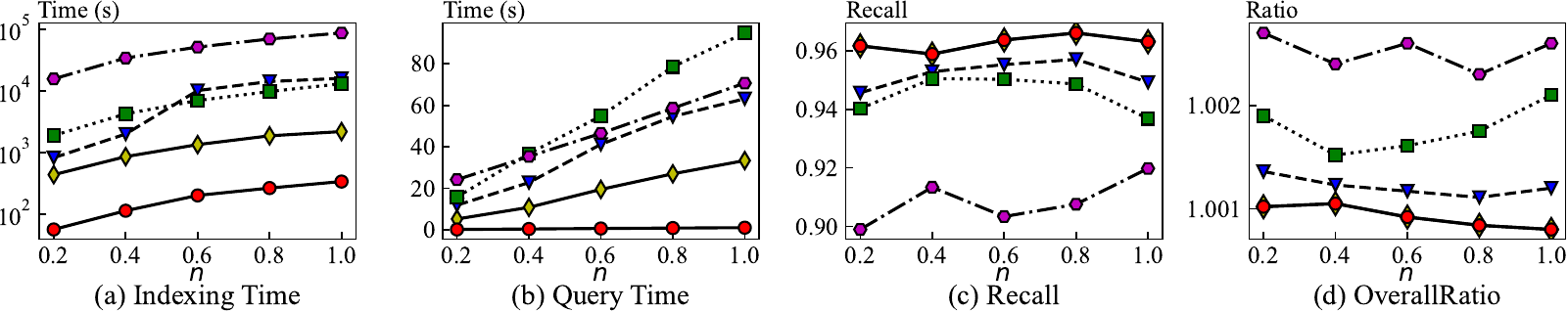}
    \vspace*{-0.1cm}
	\caption{Scalability: performance under different $n$ on Microsoft SPACEV500M.}
	\label{scalability}
\end{figure*}

\begin{figure*}[tb] 
	\centering
	\includegraphics[width=0.72\linewidth]{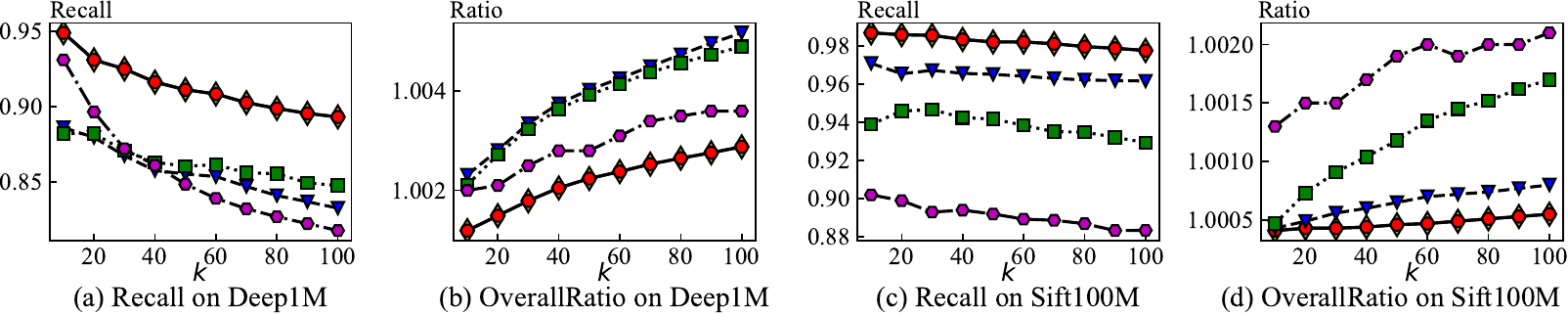}
    \vspace*{-0.1cm}
	\caption{Performance under different $k$.}
	\label{diffk}
\end{figure*}

\subsection{Comparison with LSH-based Methods} \label{comparison}

\subsubsection{Indexing Performance} 

Figure~\ref{indexsize} and Figure~\ref{indexingtime} summarize the indexing performance on all datasets, where the encoding time is included for both DET-LSH and PDET-LSH. 
Among single-threaded LSH-based methods, DET-LSH achieves the highest indexing efficiency, as DB-LSH and PM-LSH rely on costly data-oriented partitioning of the projected space, whereas DET-LSH constructs DE-Trees by independently dividing and encoding each projected dimension; in contrast, LCCS-LSH incurs substantial overhead due to building the Circular Shift Array (CSA). 
The advantage of DET-LSH grows with dataset cardinality: while it is slower than PM-LSH for $n\!\le\!1$M due to constructing four DE-Trees, it attains $2\times$--$6\times$ indexing speedups as $n$ increases from 10M to 500M, reflecting the linear construction cost of DE-Trees. 
As the only parallel method, PDET-LSH further provides significant gains, achieving up to $40\times$ indexing speedup on large datasets. 
Regarding index size, DET-LSH is less competitive on small datasets but becomes advantageous at scale since it stores only one-byte iSAX representations per point, although building four DE-Trees limits this benefit at small scale; PDET-LSH incurs additional space overhead due to per-thread index allocation.

\subsubsection{Query Performance}
We evaluate query performance using the Recall--Time and OverallRatio--Time curves in Figures~\ref{recalltime} and \ref{ratiotime}. 
DET-LSH consistently outperforms all single-threaded LSH-based methods in both efficiency and accuracy, achieving up to $2\times$ query speedup as $n$ increases, owing to the fact that nearby points share similar DE-Tree encodings and thus yield higher-quality candidates via range queries. 
Moreover, DET-LSH offers the best efficiency--accuracy trade-off among single-threaded methods, requiring the least query time to reach the same recall or overall ratio. 
PDET-LSH preserves exactly the same query accuracy as DET-LSH while outperforming all competitors, since its parallel query optimizations do not alter the search results. 
As the only parallel approach, PDET-LSH further delivers substantial efficiency gains, achieving up to $62\times$ query speedup on large datasets over existing LSH-based methods.

\subsubsection{Scalability}
A method has good scalability if it performs well on datasets of different cardinalities.
To evaluate scalability, we randomly sample subsets of varying cardinalities from the \textit{Microsoft SPACEV500M} dataset and compare the indexing and query performance of all methods under default settings, as shown in Figure~\ref{scalability}. 
While indexing and query times increase with dataset size for all methods, DET-LSH exhibits much slower growth than other single-threaded LSH-based approaches due to the efficiency of DE-Trees (Figures~\ref{scalability}(a) and \ref{scalability}(b)). 
PDET-LSH consistently maintains—and further amplifies—its advantages in both indexing and query efficiency as the dataset size grows. 
Meanwhile, recall and overall ratio remain stable across methods, since random sampling preserves the data distribution. 
Overall, DET-LSH and PDET-LSH exhibit better scalability than other methods.

\subsubsection{Effect of $k$} 
To study the impact of $k$, we evaluate all methods under varying $k$ and report only recall and overall ratio, since $k$ has negligible effect on query time and no effect on indexing time (Figure~\ref{diffk}). 
As $k$ increases, query accuracy of all methods slightly degrades because the number of candidate points remains fixed, increasing the likelihood of missing exact nearest neighbors. 
Across all settings, DET-LSH and PDET-LSH consistently achieve the best performance among all competitors.

\begin{figure}[tb] 
\vspace*{-0.3cm}
	\centering
	\includegraphics[width=0.9\linewidth]{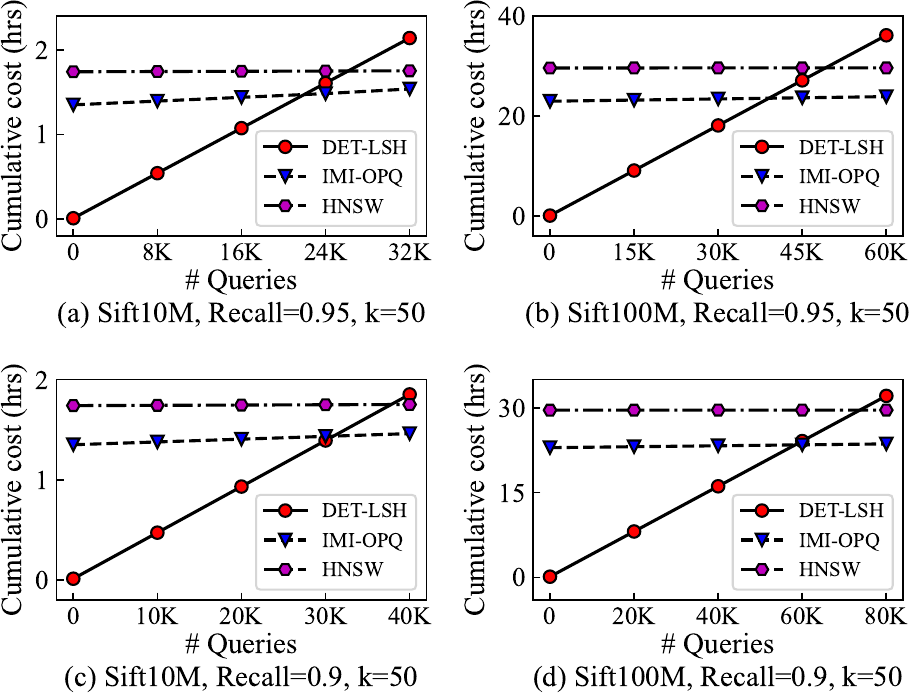}
	\caption{Cumulative query cost without parallel (start with the indexing time).}
	\label{nonlsh}
\end{figure}

\subsection{Comparison with Non-LSH-based Methods}

In this section, we compare DET-LSH and PDET-LSH with HNSW~\cite{malkov2018efficient} (graph-based method) and IMI-OPQ~\cite{ge2013optimized} (quantization-based method). 
Nevertheless, LSH-based, graph-based, and quantization-based methods have different design principles and characteristics~\cite{hydra2,li2019approximate,DBLP:journals/debu/00070P023}, making them suitable to different application scenarios. 
LSH is particularly suitable for scenarios where worst-case guarantees dominate average-case performance~\cite{indyk1998approximate,andoni2008near}, while graph and quantization are optimized for average-case efficiency under relatively stable data and query distributions~\cite{dong2011efficient}.
In particular, graph-based and quantization-based methods only support ng-approximate answers~\cite{hydra2}, that is, they do not provide any quality guarantees on their results.
It is important to emphasize that DET-LSH and PDET-LSH have to pay the cost of providing guarantees for their answers; graph-based and quantization-based methods, that do not provide any guarantees, do not have to pay this cost. 

\begin{figure}[tb] 
\vspace*{-0.3cm}
	\centering
	\includegraphics[width=0.9\linewidth]{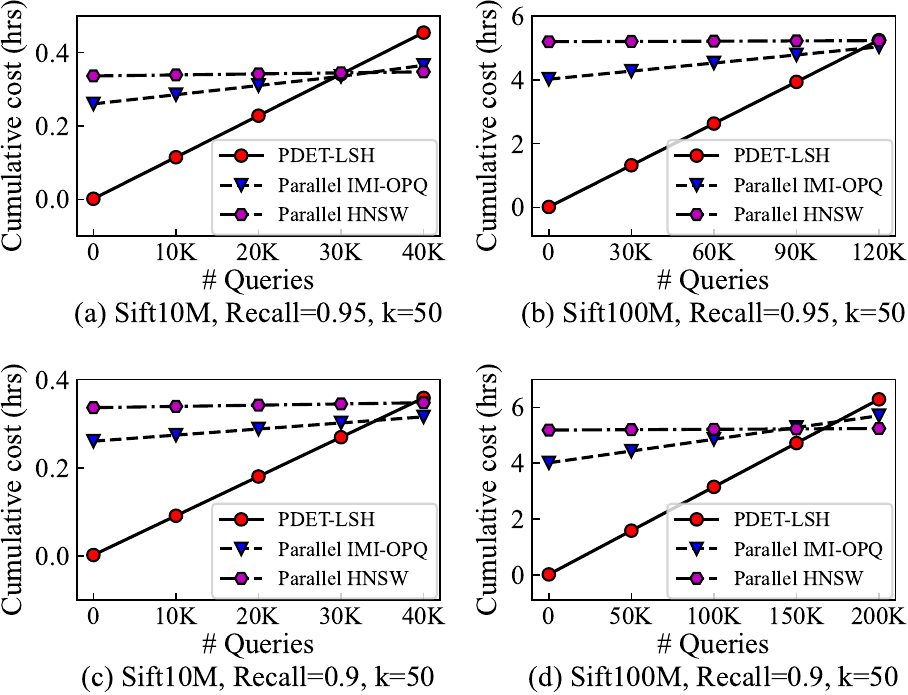}
	\caption{Cumulative query cost in parallel (start with the indexing time).} 
	\label{nonlsh_parallel}
\end{figure}

Given the considerable differences in both indexing and query efficiency across different types of ANN methods, a holistic evaluation that jointly considers indexing and query performance is necessary.
Figure~\ref{nonlsh} and Figure~\ref{nonlsh_parallel} show the cumulative query costs of all methods under the same recall, where the cost starts with the indexing time.
Experimental results demonstrate clear advantages of DET-LSH and PDET-LSH: under the same hardware configurations (single-threaded or parallel), DET-LSH/PDET-LSH completes index construction and serves up to 60K-160K queries before the best competitor answers its first query, highlighting their suitability for scenarios requiring rapid deployment and immediate query responses, such as LLM inference acceleration~\cite{chenmagicpig}.

\section{Related Work}

With the rapid growth of data volume~\cite{fu2024securing,fu2023ti,fei2023flexauth,zhang2025polaris,zhuang2026structured}, efficient management and analysis of large-scale datasets have become increasingly critical~\cite{wei2023data,wei2025dominate,li2024disauth,zhuang2024cbcms,wei2026virtuouscycleaipoweredvector,song2026csattention,tacosigmod26}. 
Approximate Nearest Neighbor (ANN) search in high-dimensional Euclidean spaces is a fundamental problem with broad applications, for which locality-sensitive hashing (LSH) methods are particularly attractive due to their strong theoretical accuracy guarantees in $c$-ANN search~\cite{e2lsh,dblsh,c2lsh,qalsh,r2lsh,vhp,lccslsh,srs,pmlsh,detlsh}. 
Following prior work, existing LSH-based approaches can be categorized into three classes: boundary-constraint (BC)-based methods~\cite{e2lsh,lsbforest,sklsh,dblsh}, collision-counting (C2)-based methods~\cite{c2lsh,qalsh,r2lsh,vhp,lccslsh}, and distance-metric (DM)-based methods~\cite{srs,pmlsh}.

Boundary-constraint (BC) methods employ $K\!\cdot\!L$ hash functions to map data points into $L$ independent $K$-dimensional projected spaces, where each point is assigned to a hash bucket defined by a $K$-dimensional hypercube; two points are considered colliding if they share a bucket in at least one of the $L$ hash tables. 
E2LSH~\cite{e2lsh}, the seminal BC approach, uses $p$-stable LSH functions~\cite{datar2004locality} but requires constructing additional hash tables as the search radius $r$ increases, resulting in prohibitive index space overhead. 
To mitigate this issue, LSB-Forest~\cite{lsbforest} indexes projected points using B-Trees, while SK-LSH~\cite{sklsh} adopts a B$^+$-Tree-based structure to improve candidate quality with lower I/O cost; however, both methods rely on heuristics and lack LSH-style theoretical guarantees. 
DB-LSH~\cite{dblsh} represents the state of the art in BC methods with rigorous guarantees, proposing a dynamic search framework built upon R$^*$-Trees.

Collision-counting (C2) methods use $K'\!\cdot\!L'$ hash functions to build $L'$ independent $K'$-dimensional hash tables, with $K' < K$ and $L' > L$, and select candidates whose collision counts exceed a threshold $t < L'$. 
C2LSH~\cite{c2lsh} introduces this paradigm by maintaining only $K'$ one-dimensional hash tables ($K'=1$) and employing \emph{virtual rehashing} to dynamically count collisions, thereby reducing index space. 
QALSH~\cite{qalsh} further improves efficiency by indexing projected points with B$^+$-Trees to avoid explicit per-dimension collision counting. 
To reduce space consumption, R2LSH~\cite{r2lsh} and VHP~\cite{vhp} extend QALSH by mapping points into multiple two-dimensional projected spaces ($K'=2$) and modeling hash buckets as virtual hyperspheres ($K'>2$), respectively. 
Finally, LCCS-LSH~\cite{lccslsh} generalizes collision counting from discrete point occurrences to the lengths of continuous co-substrings within a unified search framework.

Distance-metric (DM) methods rely on the intuition that points close to a query $q$ in the original space remain close in the projected space, and thus use $K$ hash functions to embed data into a $K$-dimensional space. 
SRS~\cite{srs} indexes projected points with an R-tree and performs exact NN search in the projected space, while PM-LSH~\cite{pmlsh} employs a PM-Tree~\cite{pmtree}-based range query strategy to improve efficiency. 
In PM-LSH, $\beta n + k$ candidates are selected according to projected-space distances, where $\beta$ is a tunable ratio ensuring search quality and $n$ denotes the dataset cardinality.

\section{Conclusions} \label{chapter7}

In this paper, we proposed DET-LSH, a novel LSH-based approach for efficient and accurate $c$-ANN query processing in high-dimensional spaces with probabilistic guarantees. 
By combining the strengths of BC and DM methods, DET-LSH supports Euclidean distance-based range queries through multiple index trees and improves query accuracy. 
We also designed DE-Tree, a dynamic encoding-based index structure that enables efficient range query processing and outperforms existing data-oriented partitioning trees, especially on large-scale datasets. 
Furthermore, we developed PDET-LSH, an in-memory parallel extension that leverages multicore CPUs to accelerate index construction and query answering. 
Experimental results verify the superiority of DET-LSH and PDET-LSH over state-of-the-art LSH-based methods in both efficiency and accuracy.

\bibliographystyle{IEEEtran}

\bibliography{IEEEabrv,ref}

\begin{thebibliography}{10}
\providecommand{\url}[1]{#1}
\csname url@samestyle\endcsname
\providecommand{\newblock}{\relax}
\providecommand{\bibinfo}[2]{#2}
\providecommand{\BIBentrySTDinterwordspacing}{\spaceskip=0pt\relax}
\providecommand{\BIBentryALTinterwordstretchfactor}{4}
\providecommand{\BIBentryALTinterwordspacing}{\spaceskip=\fontdimen2\font plus
\BIBentryALTinterwordstretchfactor\fontdimen3\font minus \fontdimen4\font\relax}
\providecommand{\BIBforeignlanguage}[2]{{%
\expandafter\ifx\csname l@#1\endcsname\relax
\typeout{** WARNING: IEEEtran.bst: No hyphenation pattern has been}%
\typeout{** loaded for the language `#1'. Using the pattern for}%
\typeout{** the default language instead.}%
\else
\language=\csname l@#1\endcsname
\fi
#2}}
\providecommand{\BIBdecl}{\relax}
\BIBdecl

\bibitem{weber1998quantitative}
R.~Weber, H.-J. Schek, and S.~Blott, ``A quantitative analysis and performance study for similarity-search methods in high-dimensional spaces,'' in \emph{VLDB}, 1998.

\bibitem{hydra2}
K.~Echihabi, K.~Zoumpatianos, T.~Palpanas, and H.~Benbrahim, ``Return of the lernaean hydra: Experimental evaluation of data series approximate similarity search,'' \emph{VLDB}, 2019.

\bibitem{DBLP:journals/debu/00070P023}
Z.~Wang, P.~Wang, T.~Palpanas, and W.~Wang, ``Graph- and tree-based indexes for high-dimensional vector similarity search: Analyses, comparisons, and future directions,'' \emph{{IEEE} Data Eng. Bull.}, 2023.

\bibitem{suco}
J.~Wei, X.~Lee, Z.~Liao, T.~Palpanas, and B.~Peng, ``Subspace collision: An efficient and accurate framework for high-dimensional approximate nearest neighbor search,'' \emph{SIGMOD}, 2025.

\bibitem{iliassigmod25}
I.~Azizi, K.~Echihabi, and T.~Palpanas, ``Graph-based vector search: An experimental evaluation of the state-of-the-art,'' \emph{{SIGMOD}}, 2025.

\bibitem{leafi}
Q.~Wang, I.~Ileana, and T.~Palpanas, ``Leafi: Data series indexes on steroids with learned filters,'' \emph{SIGMOD}, 2025.

\bibitem{dblsh}
Y.~Tian, X.~Zhao, and X.~Zhou, ``{DB-LSH 2.0: Locality-Sensitive Hashing With Query-Based Dynamic Bucketing},'' \emph{IEEE TKDE}, 2023.

\bibitem{lccslsh}
Y.~Lei, Q.~Huang, M.~Kankanhalli, and A.~K. Tung, ``Locality-sensitive hashing scheme based on longest circular co-substring,'' in \emph{SIGMOD}, 2020.

\bibitem{pmlsh}
B.~Zheng, Z.~Xi, L.~Weng, N.~Q.~V. Hung, H.~Liu, and C.~S. Jensen, ``{PM-LSH: A fast and accurate LSH framework for high-dimensional approximate NN search},'' \emph{VLDB}, 2020.

\bibitem{detlsh}
J.~Wei, B.~Peng, X.~Lee, and T.~Palpanas, ``Det-lsh: A locality-sensitive hashing scheme with dynamic encoding tree for approximate nearest neighbor search,'' \emph{VLDB}, 2024.

\bibitem{indyk1998approximate}
P.~Indyk and R.~Motwani, ``Approximate nearest neighbors: towards removing the curse of dimensionality,'' in \emph{ACM STOC}, 1998.

\bibitem{andoni2008near}
A.~Andoni and P.~Indyk, ``Near-optimal hashing algorithms for approximate nearest neighbor in high dimensions,'' \emph{Communications of the ACM}, vol.~51, no.~1, pp. 117--122, 2008.

\bibitem{dong2011efficient}
W.~Dong, C.~Moses, and K.~Li, ``Efficient k-nearest neighbor graph construction for generic similarity measures,'' in \emph{WWW}, 2011.

\bibitem{kraska2018case}
T.~Kraska, A.~Beutel, E.~H. Chi, J.~Dean, and N.~Polyzotis, ``The case for learned index structures,'' in \emph{SIGMOD}, 2018.

\bibitem{andoni2015practical}
A.~Andoni, P.~Indyk, T.~Laarhoven, I.~Razenshteyn, and L.~Schmidt, ``Practical and optimal lsh for angular distance,'' \emph{NeurIPS}, 2015.

\bibitem{chenmagicpig}
Z.~Chen, R.~Sadhukhan, Z.~Ye, Y.~Zhou, J.~Zhang, N.~Nolte, Y.~Tian, M.~Douze, L.~Bottou, Z.~Jia \emph{et~al.}, ``Magicpig: Lsh sampling for efficient llm generation,'' in \emph{ICLR}, 2025.

\bibitem{gionis1999similarity}
A.~Gionis, P.~Indyk, R.~Motwani \emph{et~al.}, ``Similarity search in high dimensions via hashing,'' in \emph{VLDB}, 1999.

\bibitem{datar2004locality}
M.~Datar, N.~Immorlica, P.~Indyk, and V.~S. Mirrokni, ``Locality-sensitive hashing scheme based on p-stable distributions,'' in \emph{Proceedings of the twentieth annual symposium on Computational geometry}, 2004.

\bibitem{e2lsh}
A.~Andoni, ``{LSH Algorithm and Implementation (E2LSH)},'' 2005.

\bibitem{lsbforest}
Y.~Tao, K.~Yi, C.~Sheng, and P.~Kalnis, ``Quality and efficiency in high dimensional nearest neighbor search,'' in \emph{ACM SIGMOD}, 2009.

\bibitem{sklsh}
Y.~Liu, J.~Cui, Z.~Huang, H.~Li, and H.~T. Shen, ``{SK-LSH: an efficient index structure for approximate nearest neighbor search},'' \emph{VLDB}, 2014.

\bibitem{c2lsh}
J.~Gan, J.~Feng, Q.~Fang, and W.~Ng, ``Locality-sensitive hashing scheme based on dynamic collision counting,'' in \emph{ACM SIGMOD}, 2012.

\bibitem{qalsh}
Q.~Huang, J.~Feng, Y.~Zhang, Q.~Fang, and W.~Ng, ``Query-aware locality-sensitive hashing for approximate nearest neighbor search,'' \emph{VLDB}, vol.~9, no.~1, pp. 1--12, 2015.

\bibitem{r2lsh}
K.~Lu and M.~Kudo, ``{R2LSH: A nearest neighbor search scheme based on two-dimensional projected spaces},'' in \emph{ICDE}.\hskip 1em plus 0.5em minus 0.4em\relax IEEE, 2020.

\bibitem{vhp}
K.~Lu, H.~Wang, W.~Wang, and M.~Kudo, ``{VHP: approximate nearest neighbor search via virtual hypersphere partitioning},'' \emph{VLDB}, 2020.

\bibitem{srs}
Y.~Sun, W.~Wang, J.~Qin, Y.~Zhang, and X.~Lin, ``{SRS: solving c-approximate nearest neighbor queries in high dimensional euclidean space with a tiny index},'' \emph{Proceedings of the VLDB Endowment}, 2014.

\bibitem{wei2023data}
J.~Wei, Y.~Li, Y.~Fu, Y.~Zhang, and X.~Li, ``Data interoperating architecture (dia): Decoupling data and applications to give back your data ownership,'' in \emph{COMPSAC}.\hskip 1em plus 0.5em minus 0.4em\relax IEEE, 2023.

\bibitem{wei2025dominate}
J.~Wei, X.~Lee, Y.~Fu, Y.~Li, and B.~Peng, ``Dominate data by yourself: a decentralized scheme for data interoperation when data is decoupled from applications,'' \emph{World Wide Web}, vol.~28, no.~3, pp. 1--39, 2025.

\bibitem{rstartree}
N.~Beckmann, H.-P. Kriegel, R.~Schneider, and B.~Seeger, ``{The R*-tree: An efficient and robust access method for points and rectangles},'' in \emph{ACM SIGMOD}, 1990.

\bibitem{pmtree}
T.~Skopal, J.~Pokorn{\`y}, and V.~Sn{\'a}{\v{s}}el, ``{Nearest Neighbours Search using the PM-tree},'' in \emph{DASFAA}.\hskip 1em plus 0.5em minus 0.4em\relax Springer, 2005.

\bibitem{rtree}
A.~Guttman, ``{R-trees: A dynamic index structure for spatial searching},'' in \emph{ACM SIGMOD}, 1984.

\bibitem{mtree}
P.~Ciaccia, M.~Patella, P.~Zezula \emph{et~al.}, ``{M-tree: An efficient access method for similarity search in metric spaces},'' in \emph{Vldb}, 1997.

\bibitem{messi}
B.~Peng, P.~Fatourou, and T.~Palpanas, ``{Messi: In-memory data series indexing},'' in \emph{ICDE}.\hskip 1em plus 0.5em minus 0.4em\relax IEEE, 2020.

\bibitem{paris+}
------, ``Paris+: Data series indexing on multi-core architectures,'' \emph{TKDE}, 2020.

\bibitem{sing}
------, ``{SING: Sequence Indexing Using GPUs},'' in \emph{ICDE}.\hskip 1em plus 0.5em minus 0.4em\relax IEEE, 2021.

\bibitem{echihabi2022hercules}
K.~Echihabi, P.~Fatourou, K.~Zoumpatianos, T.~Palpanas, and H.~Benbrahim, ``Hercules against data series similarity search,'' \emph{VLDB}, 2022.

\bibitem{dumpyos}
Z.~Wang, Q.~Wang, P.~Wang, T.~Palpanas, and W.~Wang, ``Dumpyos: {A} data-adaptive multi-ary index for scalable data series similarity search,'' \emph{{VLDB} J.}, 2024.

\bibitem{azizi2023elpis}
I.~Azizi, K.~Echihabi, and T.~Palpanas, ``Elpis: Graph-based similarity search for scalable data science,'' \emph{VLDB}, 2023.

\bibitem{chatzakis2023odyssey}
M.~Chatzakis, P.~Fatourou, E.~Kosmas, T.~Palpanas, and B.~Peng, ``Odyssey: A journey in the land of distributed data series similarity search,'' \emph{VLDB}, 2023.

\bibitem{fatourou2023fresh}
P.~Fatourou, E.~Kosmas, T.~Palpanas, and G.~Paterakis, ``Fresh: {A} lock-free data series index,'' in \emph{SRDS}, 2023.

\bibitem{zolotarev1986one}
V.~M. Zolotarev, \emph{One-dimensional stable distributions}.\hskip 1em plus 0.5em minus 0.4em\relax American Mathematical Soc., 1986, vol.~65.

\bibitem{isax}
J.~Shieh and E.~Keogh, ``{iSAX: indexing and mining terabyte sized time series},'' in \emph{ACM SIGKDD}, 2008, pp. 623--631.

\bibitem{iSAX2}
A.~Camerra, J.~Shieh, T.~Palpanas, T.~Rakthanmanon, and E.~Keogh, ``Beyond one billion time series: indexing and mining very large time series collections with isax2+,'' \emph{KAIS}, 2014.

\bibitem{malkov2018efficient}
Y.~A. Malkov and D.~A. Yashunin, ``Efficient and robust approximate nearest neighbor search using hierarchical navigable small world graphs,'' \emph{IEEE TPAMI}, 2018.

\bibitem{ge2013optimized}
T.~Ge, K.~He, Q.~Ke, and J.~Sun, ``Optimized product quantization,'' \emph{IEEE transactions on pattern analysis and machine intelligence}, 2013.

\bibitem{amd}
{Advanced Micro Devices, Inc}, ``{4th Gen AMD EPYC™ Processor Architecture},'' 2024.

\bibitem{li2019approximate}
W.~Li, Y.~Zhang, Y.~Sun, W.~Wang, M.~Li, W.~Zhang, and X.~Lin, ``Approximate nearest neighbor search on high dimensional data—experiments, analyses, and improvement,'' \emph{IEEE TKDE}, 2019.

\bibitem{fu2024securing}
Y.~Fu, X.~Lee, J.~Wei, Y.~Li, and B.~Peng, ``Securing the internet’s backbone: A blockchain-based and incentive-driven architecture for dns cache poisoning defense,'' \emph{Computer Networks}, 2024.

\bibitem{fu2023ti}
Y.~Fu, J.~Wei, Y.~Li, B.~Peng, and X.~Li, ``Ti-dns: A trusted and incentive dns resolution architecture based on blockchain,'' in \emph{TrustCom}, 2023.

\bibitem{fei2023flexauth}
Z.~Fei, Y.~Li, J.~Wei, Y.~Fu, B.~Peng, and X.~Li, ``Flexauth: A decentralized authorization system with flexible delegation,'' in \emph{TrustCom}, 2023.

\bibitem{zhang2025polaris}
A.~Zhang, X.~Lee, Z.~Zhuang, J.~Wei, Y.~Fu, and B.~Peng, ``Polaris: Cross-domain access control via verifiable identity and policy-based authorization,'' \emph{arXiv preprint arXiv:2511.22017}, 2025.

\bibitem{zhuang2026structured}
Z.~Zhuang, X.~Lee, A.~Zhang, J.~Wei, Y.~Fu, and B.~Peng, ``Structured policy modeling and context-aware generation for multi-jurisdictional compliance in global software systems,'' \emph{IST}, p. 108041, 2026.

\bibitem{li2024disauth}
Y.~Li, J.~Wei, Z.~Fei, Y.~Fu, and X.~Lee, ``Disauth: A dns-based secure authorization framework for protecting data decoupled from applications,'' \emph{Computer Networks}, 2024.

\bibitem{zhuang2024cbcms}
Z.~Zhuang, X.~Lee, J.~Wei, Y.~Fu, and A.~Zhang, ``Cbcms: a compliance management system for cross-border data transfer,'' in \emph{BigData}, 2024.

\bibitem{wei2026virtuouscycleaipoweredvector}
J.~Wei, Q.~Xu, and C.~Yang, ``{The Virtuous Cycle: AI-Powered Vector Search and Vector Search-Augmented AI},'' in \emph{ICDE}, 2026.

\bibitem{song2026csattention}
C.~Song, Z.~Peng, J.~Wei, and C.~Yang, ``{CSA}ttention: Centroid-scoring attention for accelerating {LLM} inference,'' 2026.

\bibitem{tacosigmod26}
J.~Wei, Z.~Liao, R.~Han, Q.~Xu, C.~Yang, and T.~Palpanas, ``{TaCo: Data-adaptive and Query-aware Subspace Collision for High-dimensional Approximate Nearest Neighbor Search},'' \emph{SIGMOD}, 2026.

\end{thebibliography}

\begin{IEEEbiography}[{\includegraphics[width=1in,height=1.25in,clip,keepaspectratio]{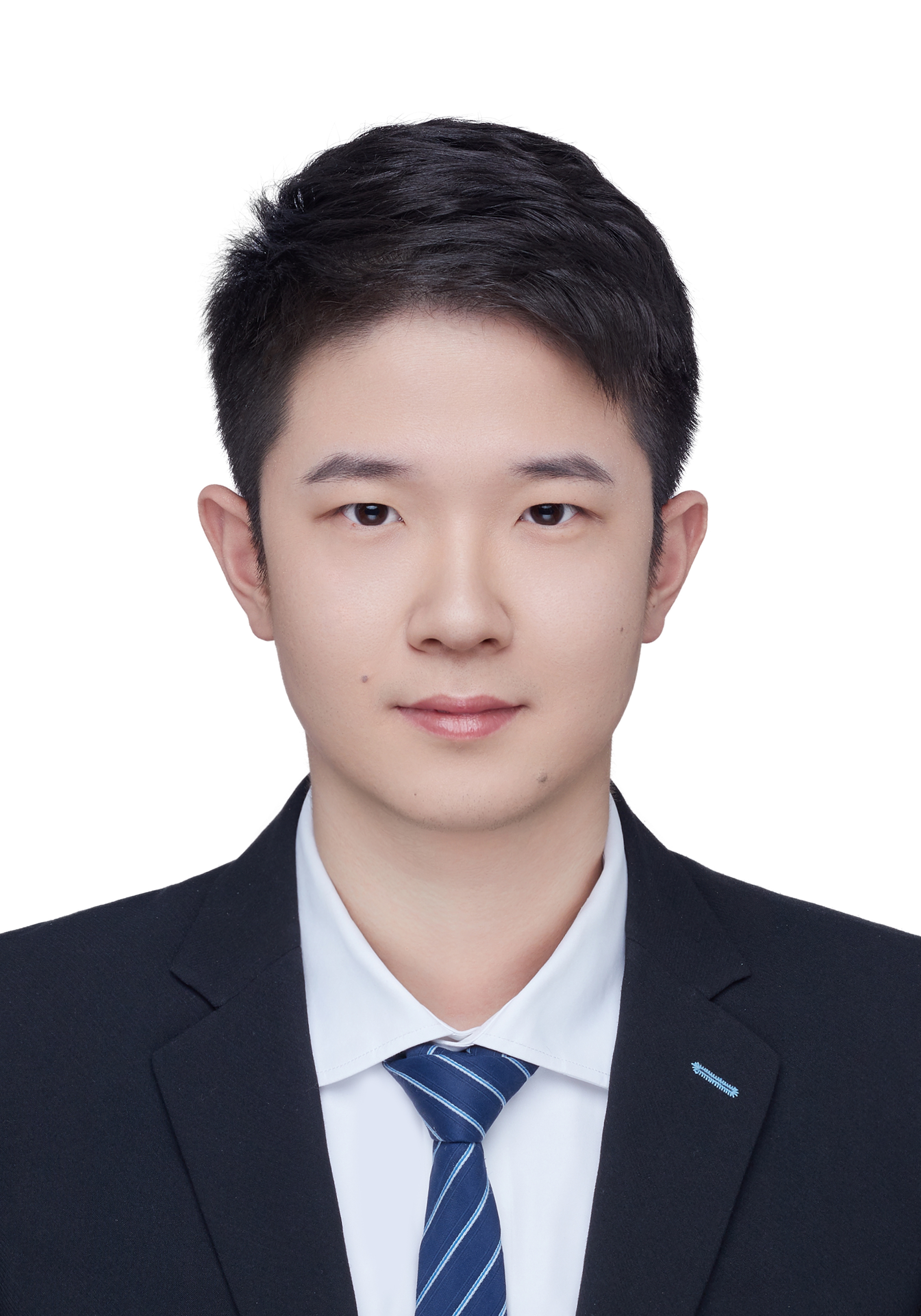}}]{Jiuqi Wei}
received his B.S. degree in software engineering from Nankai University, China, in 2019, and received his Ph.D. degree at Institute of Computing Technology, Chinese Academy of Sciences, China, in 2025. 
He is currently a researcher of Oceanbase Lab, Ant Group.
His research interests include vector database, information retrieval, LLM inference acceleration, and data management.
His research results have been published in top-tier conferences and journals such as SIGMOD, VLDB, and ICDE.
\end{IEEEbiography}

\begin{IEEEbiography}[{\includegraphics[width=1in,height=1.25in,clip,keepaspectratio]{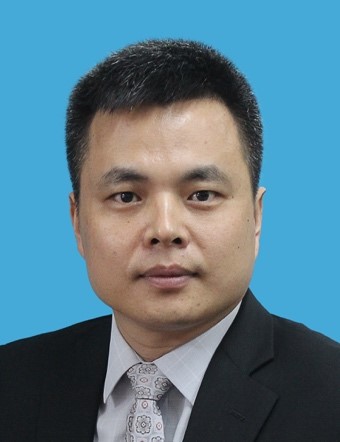}}]{Xiaodong Lee}
received his Ph.D. degree in computer science from the Institute of Computing Technology, Chinese Academy of Sciences, in 2004. He is currently a Professor at the Institute of Computing Technology, Chinese Academy of Sciences, and the Center for Internet Governance Director, Tsinghua University. 
He is also Vice Chairman of the Internet Society of China. His research interests include Internet fundamental resources management, data governance, and Internet infrastructure.
\end{IEEEbiography}

\begin{IEEEbiography}[{\includegraphics[width=1in,height=1.25in,clip,keepaspectratio]{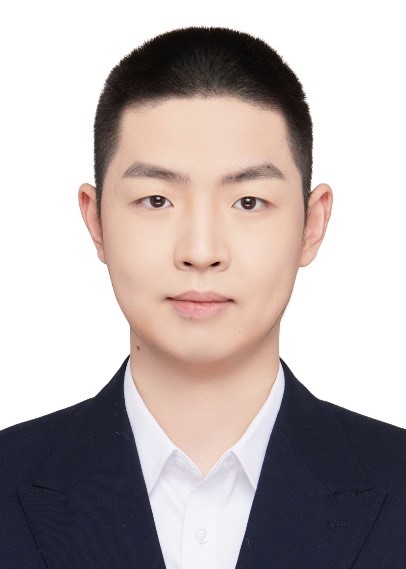}}]{Botao Peng}
received his Ph.D. degree at the Department of Computer Science, University of the Paris, in 2020, under the supervision of Themis Palpanas.
He is currently an Associate Professor at the Institute of Computing Technology, Chinese Academy of Sciences.
His research focuses on high-dimensional vector management, data analysis, and Internet infrastructure.
His research results have been published in top-tier conferences and journals such as SIGMOD, VLDB, ICDE, TKDE, and VLDBJ.
\end{IEEEbiography}

\begin{IEEEbiography}[{\includegraphics[width=1in,height=1.25in,clip,keepaspectratio]{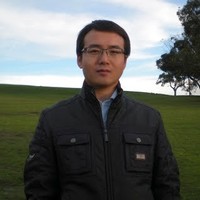}}]{Quanqing Xu}
received his Ph.D. degree in computer science from the School of Electronics Engineering and Computer Science, Peking University (PKU), in 2009.
He is currently the technical director of Oceanbase Lab, Ant Group.
His research interests primarily include distributed database systems and storage systems.
He is a Fellow of the IET, a distinguished member of CCF, a senior member of ACM and IEEE.
\end{IEEEbiography}

\begin{IEEEbiography}[{\includegraphics[width=1in,height=1.25in,clip,keepaspectratio]{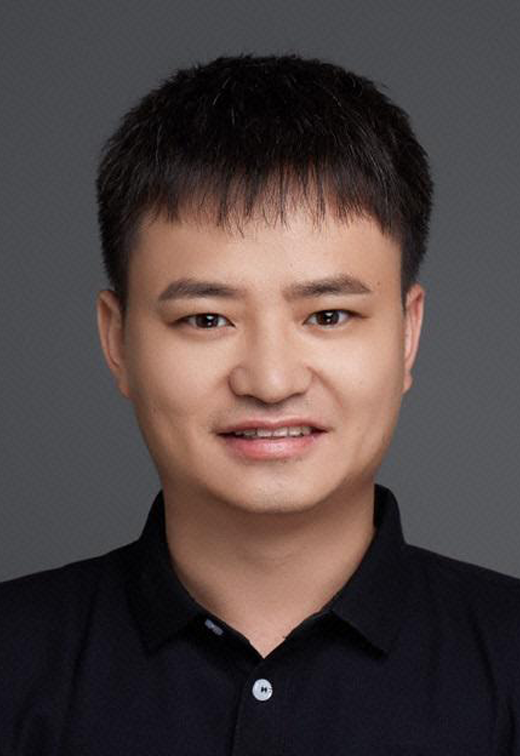}}]{Chuanhui Yang}
is the CTO of OceanBase, Ant Group.
His research focuses on database, distributed
systems, and AI infra.
As one of the founding members, he led the previous architecture design, and technology research and development of OceanBase, realizing the full implementation of OceanBase in Ant
Group from scratch.
He also led two OceanBase TPC-C tests and broke the world record, and authored the monograph "Large-Scale Distributed Storage Systems: Principles and Practice".
\end{IEEEbiography}

\begin{IEEEbiography}[{\includegraphics[width=1in,height=1.25in,clip,keepaspectratio]{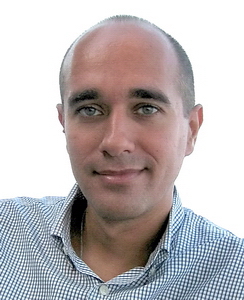}}]{Themis Palpanas}
is an ACM Fellow, a Senior Member of the French University Insitute (IUF), and a Distinguished Professor of computer science at Université Paris Cité. 
He has authored 14 patents, received 3 best paper awards and the IBM SUR award, has been Program Chair for VLDB 2025 and IEEE BigData 2023, General Chair for VLDB 2013, and has served Editor in Chief for BDR. He has been working in the fields of Data Series Management and Analytics for more than 15 years, and has developed several of the state of the art techniques.
\end{IEEEbiography}

\end{document}